%% file: main.tex
\documentclass[conference]{IEEEtran}
\usepackage[available,functional,reproduced]{ndssbadges}

\usepackage{times}
\usepackage[utf8]{inputenc}
\usepackage{url}
\usepackage{amssymb}
\usepackage{amsmath}
\usepackage{array,multirow}
\usepackage{xspace}
\usepackage{xcolor}
\usepackage{color, colortbl} 
\usepackage{balance}
\usepackage{paralist}
\usepackage[
  colorlinks = true,
  linkcolor = blue,
  urlcolor  = blue,
  citecolor = blue,
  anchorcolor = blue]{hyperref}
\usepackage{cleveref}
\usepackage{graphicx}
\usepackage{pifont}
\usepackage{array}
\usepackage{booktabs}
\usepackage{tikz}
\usepackage{bm}
\usepackage{todonotes}
\usepackage{arydshln}
\usepackage{makecell}
\usepackage{tabularx}
\usepackage{caption}
\usepackage{soul,color}
\usepackage{mwe}
\usepackage{ifthen}
\usepackage{algorithm}
\usepackage[noend]{algpseudocode}
\usepackage{epsfig}
\usepackage[tight,footnotesize]{subfigure}
\usepackage{algorithm}%
\usepackage{algpseudocode}%
\usepackage{caption}
\usepackage{cite}
\usepackage{listings}

\usepackage{amsthm}
\newtheorem{theorem}{Theorem}
\usepackage{bbm}

\usepackage{noto-emoji-easy}

\usetikzlibrary{decorations.pathmorphing, decorations.pathreplacing, decorations.shapes}
\usetikzlibrary{calc}
\usetikzlibrary{shapes.callouts}

\usepackage{dirtree}

\begin{document}
\title{Content Censorship in the\\InterPlanetary File System}

\author{
  \IEEEauthorblockN{Srivatsan Sridhar\IEEEauthorrefmark{1}, Onur Ascigil\IEEEauthorrefmark{2}, Navin Keizer\IEEEauthorrefmark{5}, François Genon\IEEEauthorrefmark{3}\\ Sébastien Pierre\IEEEauthorrefmark{3}, Yiannis Psaras\IEEEauthorrefmark{6}, Etienne Rivière\IEEEauthorrefmark{3}, Michał Król\IEEEauthorrefmark{4}}
  
  \vspace{+0.06in}

  \IEEEauthorblockA{
    \begin{tabular}{ccc}
      \IEEEauthorrefmark{1}Stanford University & \IEEEauthorrefmark{2}Lancaster University & \IEEEauthorrefmark{5}University College London \\
      \IEEEauthorrefmark{3}ICTEAM, UCLouvain & \IEEEauthorrefmark{6}Protocol Labs & \IEEEauthorrefmark{4}City, University of London \\
    \end{tabular}
  }
}

\IEEEoverridecommandlockouts
\makeatletter\def\@IEEEpubidpullup{6.5\baselineskip}\makeatother
\IEEEpubid{\parbox{\columnwidth}{
    Network and Distributed System Security (NDSS) Symposium 2024\\
    26 February - 1 March 2024, San Diego, CA, USA\\
    ISBN 1-891562-93-2\\
    https://dx.doi.org/10.14722/ndss.2024.23153\\
    www.ndss-symposium.org
}
\hspace{\columnsep}\makebox[\columnwidth]{}}

\maketitle

\input{macros}
\begin{abstract}
  The InterPlanetary File System (IPFS) is currently the largest decentralized storage solution in operation, with thousands of active participants and millions of daily content transfers. IPFS is used as remote data storage for numerous blockchain-based smart contracts, Non-Fungible Tokens (NFT), and decentralized applications.

  We present a content censorship attack that can be executed with minimal effort and cost, and that prevents the retrieval of any chosen content in the IPFS network. 
  The attack exploits a conceptual issue in a core component of IPFS, the Kademlia Distributed Hash Table (DHT), which is used to resolve content IDs to peer addresses.
  We provide efficient detection and mitigation mechanisms for this vulnerability. 
  Our mechanisms achieve a 99.6\% detection rate and mitigate 100\% of the detected attacks with minimal signaling and computational overhead.
  We followed responsible disclosure procedures, 
  and our countermeasures are scheduled for deployment in the future versions of IPFS.
\end{abstract}

\input{sections/intro}

\input{sections/ethics}

\input{sections/background}

\input{sections/model}

\input{sections/attack}

\input{sections/detection}

\input{sections/mitigation}

\input{sections/evaluation}

\input{sections/related}

\input{sections/discussion}

\input{sections/conclusion}

\section*{Acknowledgements}
This work was partly done when Srivatsan Sridhar was consulting for Protocol Labs. At Stanford University, Srivatsan Sridhar's research is funded by a gift from the Ethereum Foundation.

\bibliographystyle{IEEEtranS}
\bibliography{refs}

\clearpage

\appendices
\crefalias{section}{appendix}
\crefalias{subsection}{subappendix}
\crefalias{subsubsection}{subsubappendix}

\input{sections/ae_appendix.tex}

\end{document}

%% file: macros.tex
\newboolean{showcomments}
\setboolean{showcomments}{true}

\ifthenelse{\boolean{showcomments}}
{ \newcommand{\mynote}[3]{
    \protect\fbox{\bfseries\sffamily\scriptsize#1}
    {\small\textsf{\emph{\color{#3}{#2}}}}}}
{ \newcommand{\mynote}[3]{}}

\newcommand{\er}[1]{\mynote{Etienne}{#1}{blue}}
\newcommand{\michal}[1]{\mynote{Michał}{#1}{brown}}
\newcommand{\yiannis}[1]{\mynote{Yiannis}{#1}{orange}}
\newcommand{\srivatsan}[1]{\mynote{Srivatsan}{#1}{green}}
\newcommand{\navin}[1]{\mynote{Navin}{#1}{violet}}
\newcommand{\onur}[1]{\mynote{Onur}{#1}{red}}

\newcommand\attacker{attacker\xspace}
\newcommand\adversary{attacker\xspace}
\newcommand\prerecorded{prerecorded\xspace}
\newcommand\prerecords{prerecords\xspace}
\newcommand\prerecord{prerecord\xspace}
\newcommand\bitswap{\emph{Bitswap}\xspace}

\newcommand{\cf}{cf.\@\xspace}
\newcommand{\vs}{vs.\@\xspace}
\newcommand{\etc}{etc.\@\xspace}
\newcommand{\ala}{ala\@\xspace}
\newcommand{\wrt}{w.r.t.\@\xspace}
\newcommand{\etal}{\textit{et al.}\@\xspace}
\newcommand{\eg}{\textit{e.g.}\@\xspace}
\newcommand{\Eg}{\textit{E.g.}\@\xspace}
\newcommand{\ie}{\textit{i.e.}\@\xspace}
\newcommand{\Ie}{\textit{I.e.}\@\xspace}
\newcommand{\via}{\textit{via}\@\xspace}
\newcommand{\defacto}{\textit{de facto}\@\xspace}

\newcommand\mypara[1]{\vspace{0.05in} \noindent \textbf{#1.}}
\newcommand\para[1]{\vspace{0.05in} \noindent \textbf{#1.}}

\newcommand\definition[2]{\ding{118}\xspace \textsf{#1}\xspace$\bm{\rightarrow}$\xspace(#2):\xspace}

\def\downloader{downloader\xspace}
\def\Downloader{Downloader\xspace}
\def\downloaders{downloaders\xspace}
\def\Downloaders{Downloaders\xspace}
\def\provider{provider\xspace}
\def\providers{providers\xspace}
\def\resolver{resolver\xspace}
\def\resolvers{resolvers\xspace}

\newcommand\inlinesection[1]{{\bf #1.}}

\def\first{({\it i})\xspace}
\def\second{({\it ii})\xspace}
\def\third{({\it iii})\xspace}
\def\fourth{({\it iv})\xspace}
\def\fifth{({\it v})\xspace}
\def\sixth{({\it vi})\xspace}

\newcommand{\one}{({\em i})\xspace}
\newcommand{\two}{({\em ii})\xspace}
\newcommand{\three}{({\em iii})\xspace}
\newcommand{\four}{({\em iv})\xspace}
\newcommand{\five}{({\em v})\xspace}

\definecolor{verylightgray}{gray}{0.8}

\newcolumntype{L}{l<{\hspace{1cm}}}
\newcolumntype{C}{c<{\hspace{1cm}}}
\newcolumntype{D}{c<{\hspace{0.3cm}}}

\newcommand\vgap{\vskip 2ex}
\newcommand\marker{\vgap\ding{118}\xspace}

\def\na{--}
\def\unsure{?}
\def\missing{$!$}
\newcommand{\yes}{\ding{51}}
\newcommand{\no}{\ding{55}}
\DeclareRobustCommand\pie[1]{
\tikz[every node/.style={inner sep=0,outer sep=0, scale=1.5}]{
\node[minimum size=1.5ex] at (0,-1.5ex) {}; 
 \draw[fill=white] (0,-1.5ex) circle (0.75ex); \draw[fill=black] (0.75ex,-1.5ex) arc (0:#1:0.75ex); 
}
}
\def\L{\pie{0}} %
\def\M{\pie{-180}} %
\def\H{\pie{360}} %

\newcommand{\cmark}{\ding{51}}%
\newcommand{\xmark}{\ding{55}}%

\newcommand{\peerid}{\ensuremath{\mathsf{peerid}}}
\newcommand{\cid}{\ensuremath{\mathsf{cid}}}
\newcommand{\key}{\ensuremath{\mathsf{key}}}
\newcommand{\threshold}{\ensuremath{\mathsf{thr}}}

\newcommand{\numSybils}{\ensuremath{e}}
\newcommand{\cgen}{\ensuremath{c_{\mathrm{gen}}}}
\newcommand{\twarmup}{\ensuremath{t_{\mathrm{w}}}}
\newcommand{\teff}{\ensuremath{t_{\mathrm{eff}}}}
\newcommand{\coper}{\ensuremath{c_{\mathrm{oper}}}}
\newcommand{\catt}{\ensuremath{c_{\mathrm{att}}}}
\newcommand{\aeff}{\ensuremath{a_{\mathrm{eff}}}}
\newcommand{\meff}{\ensuremath{m_{\mathrm{eff}}}}

\newcommand{\algvar}[1]{\ensuremath{\mathsf{#1}}}
\newcommand{\CPL}{\algvar{CPL}}

\newcommand{\vspacebeforesection}{}
\newcommand{\vspaceaftercaption}{}

\newcommand{\artifactDOI}{10.5281/zenodo.8300034}
\newcommand{\artifactlink}{https://doi.org/10.5281/zenodo.8300034}
\newcommand{\githublink}{https://github.com/ssrivatsan97/go-libp2p-kad-dht}

%% file: sections/intro.tex
\vspacebeforesection
\section{Introduction}
Inter-Planetary FileSystem (IPFS) is the largest decentralized peer-to-peer filesystem currently in operation. The platform underpins various decentralized web applications~\cite{ecosystem}, including social networking and discussion (Discussify~\cite{discussify}, Matters News~\cite{matters}), data storage (Space~\cite{Space}, Peergos~\cite{Peergos}, Temporal~\cite{temporal}), content search (Almonit~\cite{almonit}, Deece~\cite{deece}), messaging (Berty~\cite{Berty}), content streaming (Audius~\cite{Audius}, Watchit~\cite{Watchit}, DTube~\cite{dtube}), gaming (Gala~\cite{gala}, Splinterlands~\cite{Splinterlands}), and e-commerce (Ethlance~\cite{ethlance}, dClimate~\cite{dClimate}). IPFS is widely used as external storage for blockchain-based applications, including valuable NFT platforms. Support for accessing IPFS has further been integrated into HTTP gateways (\eg, Cloudflare) and mainstream browsers such as Opera and Brave, allowing easy uptake. 
The IPFS network currently contains a steady number of 25,000 online nodes, spread across 2,700 Autonomous Systems and 152 countries, according to a recent study~\cite{trautwein2022design} that also observed widespread usage by clients with 7.1 million content retrieval operations observed from a single vantage point and during a single day.

IPFS is a content-centric network where each piece of content is identified by a Content Identifier (CID), similarly to BitTorrent~\cite{bittorrent} or Content-Centric Networking~\cite{zhang2014named}.
CIDs are derived by hashing the content and do not embed any network location information. %
Such an approach enables easy content deduplication and the retrieval of data from the closest available location, in addition to maintaining data integrity.
Data retrieval in a content-centric network requires mapping CIDs into network identifiers (\ie IP addresses and port numbers) of nodes hosting the content, called \emph{\providers}. 
Without this resolution mechanism, nodes willing to fetch data, or \emph{\downloaders}, have no means to know where to send their requests for data. 
The design of IPFS results from decades of research on how to build efficient P2P systems~\cite{lua2005survey,androutsellis2004survey}.
It uses resolution based on a Distributed Hash Table (DHT) combined with Bitswap, a flooding-based, unstructured search mechanism. Similarly to systems such as Gnutella~\cite{ripeanu2001peer}, \downloaders use Bitswap to establish connections to random peers in the network and send them content queries. Bitswap acts as a lightweight cache and speeds up the retrieval of popular content, but cannot provide discovery guarantees, in particular for newer or less popular data.

Reliable content discovery is provided by the DHT-based resolution system. Nodes hosting content advertise themselves as \providers in the network. First, they create \emph{provider records} linking their hosted content (identified by CIDs) to their network location (\ie, IP address and port number). Second, the \providers send the provider records to be stored on a fixed number of designated nodes. We refer to those nodes as \emph{\resolvers}.  
\Downloaders wishing to fetch the content contact the same \resolvers, retrieve the relevant provider records, and then directly contact the discovered \providers to download the data. The DHT guarantees to find the content if it is stored in the network.
Its proper operation is, therefore, of paramount importance to ensure content availability.
IPFS uses the \texttt{libp2p} implementation~\cite{libp2p_github} of the Kademlia DHT~\cite{maymounkov2002kademlia}.

\para{Contributions}
We make four main contributions.

First, we present a content censorship attack targeting the main IPFS DHT-based resolution system.
The attack relies on strategically placing Sybil identities in the network so that they replace honest \resolvers for a given CID.
As a result, \downloaders cannot discover \provider records for the target CID and are unable to download the content.
The attack can be performed from a single, resource-constrained machine at very little cost (\$4 using AWS) and makes the \provider records unavailable after a time that ranges from a few seconds to up to 48h depending on the initial setup.
Currently, IPFS has no mechanisms to counter the attack, threatening the 
security of systems using IPFS as a storage platform.
This includes collaborative file hosting solutions such as Filecoin~\cite{psaras2020interplanetary} and systems building upon it~\cite{huang2020secure, de2021accelerating}.
It also concerns the many proposals combining IPFS for storage with blockchain-hosted application logic, e.g., to implement social networks~\cite{xu2018building}, domain-specific data sharing applications~\cite{jianjun2020research,mukne2019land}, or decentralized equivalents to centralized services such as ride-sharing~\cite{hossan2021securing}.

Second, we present a reliable attack detection technique that analyzes the distribution of peer IDs in the network using the KL Divergence metric~\cite{g-test}. 
This method extends previous work~\cite{cholez2010detection} and leverages a local density-based network size estimator~\cite{eli-sohl-dht-size-estimation,kostoulas2005decentralized, manku2003symphony} to automatically adjust the detection to the dynamic size of the IPFS network.
The detection allows \providers to execute mitigation techniques, which may be more costly than the default mode, only when an attack is detected, thereby  minimizing the overhead when there is no attack. 
The detection can be performed by any node during regular content resolution operations and does not incur any additional message overhead.
In our experiments on the live IPFS network, our attack detection method was able to detect 99\% of the attacks with a false positive rate of 4\%, while allowing users to trade off these rates based on individual preferences.
A higher detection rate ensures better security while also leading to more false positives that increase the overhead when there is no attack.

Third, we introduce a mitigation technique that allows us to reliably discover provider records regardless of the number of Sybil nodes placed by an attacker around the target CID.
The mitigation replaces the regular \emph{put} and \emph{get} DHT operation by hash space region-based queries. 
Using these, \providers always find honest \resolvers to store their provider records and  querying nodes always discover these honest \resolvers and receive true provider records. While introducing an overhead sub-linear
in the number of Sybil nodes placed close to the target CID, this mitigation is only enforced when suspicions exist about the existence of an attack, as indicated by our detection mechanism.

Finally, we implement the attack using a custom IPFS DHT server node, and we implement our detection and mitigation techniques on top of the \texttt{libp2p} DHT~\cite{libp2p_github}.\footnote{See \Cref{sec:ae-appendix} for details including where to find the implementations.}
Importantly, the detection and mitigation implementations are fully compatible with the unmodified IPFS clients and can be incrementally deployed in the system. 
Therefore, while nodes that have not upgraded may remain vulnerable to the censorship attack, they continue to interoperate with nodes that have upgraded.
We evaluate the feasibility of the attack and the efficiency of the countermeasures using simulations as well as actual experiments on the live IPFS network.
Our proposed detection and mitigation schemes are scheduled to be deployed in the next release of the \texttt{libp2p} DHT.

\para{Related Attacks and Mitigations}
Similar DHT vulnerabilities have been previously discussed in the literature~\cite{dabek2001wide} and multiple prevention mechanisms have been proposed~\cite{dabek2001wide, danezis2009sybilinfer, dan2012centralized, danezis2005sybil}. However, mostly due to practical reasons~\cite{dabek2001wide, dan2012centralized, baumgart2007s} or unrealistic assumptions~\cite{danezis2009sybilinfer, danezis2005sybil, prunster2018holistic}, these mechanisms cannot be deployed in modern decentralized systems. We provide a detailed discussion on this topic in \Cref{sec:related}. As a result, multiple top-tier systems currently rely on a vulnerable DHT for various purposes.
The peer discovery mechanism for several blockchains (\eg, Ethereum~\cite{eth19discovery}, Celestia~\cite{celestia}, or Polkadot~\cite{burdges2020overview}) uses the same \texttt{libp2p} DHT implementation as IPFS. File sharing in I2P~\cite{timpanaro2015evaluation} and data dissemination in Dat~\cite{dat23dat} also use the Kademlia DHT, although with a different implementation.
Our contributions (both the attack, its detection, and mitigation mechanisms) are expected to apply to these systems as well and more broadly to systems using Kademlia or a Kademlia-like DHT.

\para{Outline}
In \Cref{sec:ethics}, we discuss ethical considerations for our study. \Cref{sec:background} presents background on IPFS and its content resolution mechanisms. \Cref{sec:attack}, \Cref{sec:detection}, and \Cref{sec:mitigation} respectively introduce the attack, its detection, and mitigation techniques. 
In \Cref{sec:evaluation} we evaluate all these mechanisms experimentally and we discuss related work in \Cref{sec:related}. \Cref{sec:discussion} provides a discussion on the implication of the attack and its countermeasures, while \Cref{sec:conclusion} concludes the paper.
This paper has an accompanying artifact which contains implementations of the attack, detection and mitigation, and experiments. \Cref{sec:ae-appendix} describes how to access the artifact and run the experiments to reproduce the results stated in this paper.

%% file: sections/ethics.tex
\vspacebeforesection
\section{Ethical Considerations and Responsible Disclosure}\label{sec:ethics}

Our work discovers a vulnerability in an existing system with thousands of users, and our experiments involve mounting an attack on the live IPFS network. This may raise ethical concerns. However, we have worked closely with Protocol Labs, the company that created and maintains IPFS and \texttt{libp2p}, to develop mitigations for the attack and ensure that our experiments do not cause harm to any existing users.

As soon as we discovered the potential security impact of the vulnerability in June 2021, we initiated a responsible disclosure process by contacting Protocol Labs. Subsequently, the vulnerability has been assigned a CVE record CVE-2023-262481\footnote{\url{https://cve.mitre.org/cgi-bin/cvename.cgi?name=CVE-2023-26248}}. Details of the CVE record will be made public once the mitigation is deployed and this work is published.
To ensure no harm or danger to regular operations or honest users, we limited our censorship attacks to randomly generated content published by our own nodes. Our attacker nodes were implemented to only drop messages related to our randomly generated content and behave normally otherwise to avoid any other side effects on the network. 

Throughout the process, Protocol Labs actively supported our research by providing access to their datasets, including network crawls. They welcomed further research based on our findings and provided advice on performance requirements for detection and mitigation techniques. 
Public disclosure of the vulnerability was carried out at IPFS Camp 2022~\cite{ipfs_camp2022}
held by Protocol Labs,
which included discussions with IPFS developers, maintainers, and other researchers.
We are currently working with Protocol Labs on the deployment of our detection and mitigation methods in the next version of the \texttt{libp2p} DHT.
This deployment will benefit other systems using \texttt{libp2p}, including Ethereum ~\cite{eth19discovery}, Celestia~\cite{celestia}, or Polkadot~\cite{burdges2020overview}.
Our mechanisms are not specific to the \texttt{libp2p} implementation of the Kademlia protocol and can be ported to other implementations of the DHT by their maintainers.
We commit to assisting these maintainers in this task.

%% file: sections/background.tex
\vspacebeforesection
\section{Background}
\label{sec:background}

In this section, we provide the necessary background information to ensure a comprehensive understanding of the attack described in this paper. We start with a description of the Distributed Hash Table (DHT) used by IPFS, followed by its content resolution mechanisms. We also detail techniques for network size estimation, necessary for our attack detection and mitigation mechanisms.

\vspacebeforesection
\subsection{IPFS DHT}
\label{sec:kad_dht}

We review the features of the Kademlia DHT~\cite{maymounkov2002kademlia} and its \texttt{libp2p} implementation~\cite{libp2p_github} that are the most relevant to our attack.
To participate in the DHT, each peer generates a public/private key pair and derives an identity $\peerid \in \{0,1\}^{256}$ as the hash of its public key.
Ideally, each peer generates a random key pair and, therefore, peer IDs are distributed uniformly and independently over the space $\{0,1\}^{256}$.
While honest nodes follow this rule, malicious nodes may generate and choose from an arbitrary number of key pairs.
Each peer maintains a routing table consisting of $m=256$ buckets.
The $i$-th bucket contains the addresses of up to $k=20$ peers whose peer IDs share a common prefix of exactly $i$ bits with the peer's own peer ID. 

A new participant node joins the IPFS network by contacting one of the hardcoded bootstrap nodes. This bootstrap node provides the new node with some initial peers allowing it to join the DHT. The new node uses this information to perform a walk through the DHT towards its own peer ID.
The walk allows to: \textit{(i)}~make sure that there is no other node in the network with the same ID; \textit{(ii)}~discover new peers and fill the newcomer's DHT routing table. At the same time, the newcomer establishes \bitswap~\cite{de2021accelerating} connections to a subset of encountered peers (usually around 300 of them). The core role of the \bitswap protocol is to enable bilateral content transfer and to play the role of a cache for recently-accessed content.

The main DHT operation $\Call{GetClosestPeers}{\key}$ returns the $k=20$ closest peers to $\key$. 
In Kademlia, the distance between two keys $x$ and $y$ in the key space is given by $x \oplus y \in \{0,...,2^{256}-1\}$, where $\oplus$ denotes the bitwise XOR operation on the keys; the resulting binary string is interpreted as an integer.
When a client wants to find the peers with IDs closest to $\key$, it sends a request to the $\alpha=3$ peers in its routing table whose peer IDs are closest to $\key$. Each of these peers returns the $k$ closest peers to $\key$ in its own routing table and the addresses of these peers. 
The client again sends a request to the $\alpha$ peers closest to $\key$, among peers in its routing table and those whose addresses it just received. This process repeats until the client does not find any more peers closer to $\key$.
Due to network churn and imperfect routing tables, we observed in our experiments that successive calls to $\Call{GetClosestPeers}{\key}$ do not always return the true set of $k=20$ closest peers (we provide more details in \Cref{sec:evaluation}, \Cref{fig:20closest}).

\vspacebeforesection
\subsection{Content Resolution in IPFS}
\label{sec:ipfs}

IPFS is a content-centric network.
It allows its participant to request files without specifying their location. 
Content is indexed by content IDs $\cid \in \{0,1\}^{256}$ that are derived from a hash of that content.
Both peer IDs and CIDs are used as keys in the DHT.
Each node can play the role of a \provider, \downloader, or \resolver. 
The process of content advertisement and resolution is illustrated in \Cref{fig:add_get_provider}.

When a \provider wishes to publish content with a given $\cid$ on IPFS, it creates a \emph{provider record} that contains $cid$ and the \provider's address.
During a $\Call{Provide}{\cid}$ operation, the \provider first uses $\Call{GetClosestPeers}{\cid}$ to locate the $k=20$ peers with their peer IDs closest to $\cid$, 
and then sends them a $\mathsf{PutProvider}$ message including the provider record (\Cref{fig:add_get_provider}(a)).
We call the peers that hold provider records for $\cid$ the \emph{resolvers} for $\cid$.

Each CID can have several \providers. In fact, by default, each IPFS client becomes a provider for each piece of content it downloads for a fixed amount of time (12h, 24h, or 48h depending on the client version or custom configuration). As a result, the system provides an auto-scaling feature with supply automatically rising with demand.

When a \downloader wishes to fetch a piece of content, it first sends a request to all its \bitswap peers. If none of them has the content, the \downloader uses the DHT-based resolution system. We stress that the \bitswap protocol plays the supporting role of a cache in the dissemination of popular files. However, the mechanism does not provide reliable content resolution, in particular for new or less popular content. %

When \bitswap unstructured search fails, the \downloader resolves $\cid$ using $\Call{FindProviders}{\cid}$. This operation uses a DHT walk identical to that of $\Call{GetClosestPeers}{\cid}$ to find $k$ \resolvers but also queries encountered nodes for a provider record for $\cid$ (\Cref{fig:add_get_provider}(b)). The process terminates when either 20 \providers have been found, or all \resolvers have been asked. Querying all encountered nodes (\ie, not only the designated \resolvers) is useful because some of the encountered nodes may have a provider record in their cache.

Upon receiving a provider record, the client connects to the address specified in the provider record to retrieve the actual content (\Cref{fig:add_get_provider}(c)).
Provider records are not authenticated, and therefore malicious \providers may respond with incorrect provider records (or may not respond at all). However, the integrity of the content is preserved because the hash of the retrieved content can be verified against its $\cid$.

\input{img/add_get_provider.tex}

\vspacebeforesection
\subsection{Network Size Estimator}
\label{sec:netsize}

The number of nodes in a decentralized system is generally unknown due to the avoidance of centralized membership management.
This number is nonetheless useful for optimizations, deciding on individual node configurations, or security mechanisms.
Various methods were proposed for the decentralized estimation of unstructured and structured networks~\cite{eli-sohl-dht-size-estimation,kostoulas2005decentralized, manku2003symphony}.
We use in this work a mechanism developed initially by Protocol Labs as part of a mechanism for decreasing the latency of publishing content in IPFS~\cite{network-size-estimation-notion,network-size-estimation-github-pr}.

Each node in the DHT refreshes its routing table periodically (every $10$ minutes in \texttt{libp2p}). 
For this, the node samples $m$ random keys (one for each bucket of its routing table)
and queries the DHT to obtain the $k=20$ closest peer IDs to each key.
Using these, the node then computes the average distance between each one of these keys $\key_j$ for $j=1,\dots,m$ and their $i$-th closest peer ID for $i=1,...,k$ (with $m=256$ and $k=20$).
\begin{equation}
    \label{equ:avg-dist}
    \overline{D}_i = \frac{1}{m} \sum_{j=1}^m \operatorname{dist}(\key_j, \peerid_{j}^{(i)})
\end{equation}
where $\peerid_{j}^{(i)}$ is the $i$-th closest peer ID to $\key_j$.
With $N$ peers in the DHT and peer IDs uniformly distributed in the hash space, the expected distance between a $\key$ and its $i$-th closest peer ID is $\frac{2^{256}i}{N+1}$. The node then runs a least square regression to compute the value of $N$ for which the expected distances best fit the empirical average distances, \ie,
\begin{equation}
    \label{equ:netsize-least-squares}
    \hat{N} = \arg\min_{N} \sum_{i=1}^k \left(\overline{D}_i - \frac{2^{256}i}{N+1}\right)^2.
\end{equation}
The resulting estimate $\hat{N}$ can be computed in closed form.

When a node starts running, it must perform DHT queries for a few random keys to initialize its network size estimate. 
Since a larger number of queries will result in higher accuracy, making more queries than what is needed to initialize one's routing table is recommended.
Thereafter, keeping the estimate up-to-date does not require any excess DHT queries beyond what is already used for refreshing the routing table as this is done frequently (every 10 minutes).

While the network size estimate has a stochastic variance resulting from the probability distribution of the honest peer IDs, it is hard for an attacker to bias the estimate significantly. Since the estimator uses the density of peer IDs around keys chosen uniformly at random, the adversary would require numerous Sybil nodes (on the order of the whole network size) to significantly affect the peer ID density around those keys.

%% file: img/add_get_provider.tex
\begin{figure}[tb]%
    \centering%
    \begin{tikzpicture}[]%
        \footnotesize%
        \draw[->] (0,0) -- (8,0) node[anchor=south] {keyspace};%
        \draw (0.5,0) node[circle, draw, fill=white, minimum size=3pt, inner sep=1pt] (p1) {1};%
        \draw (1.5,0) node[circle, draw, fill=white, minimum size=3pt, inner sep=1pt] (p2) {2};%
        \draw (2.5,0) node[circle, draw, fill=white, minimum size=3pt, inner sep=1pt] (p3) {3};%
        \draw (3.0,0) node[circle, draw, fill=white, minimum size=3pt, inner sep=1pt] (p4) {4};%
        \node at (4,0) (cid) {\includegraphics[width=0.5cm]{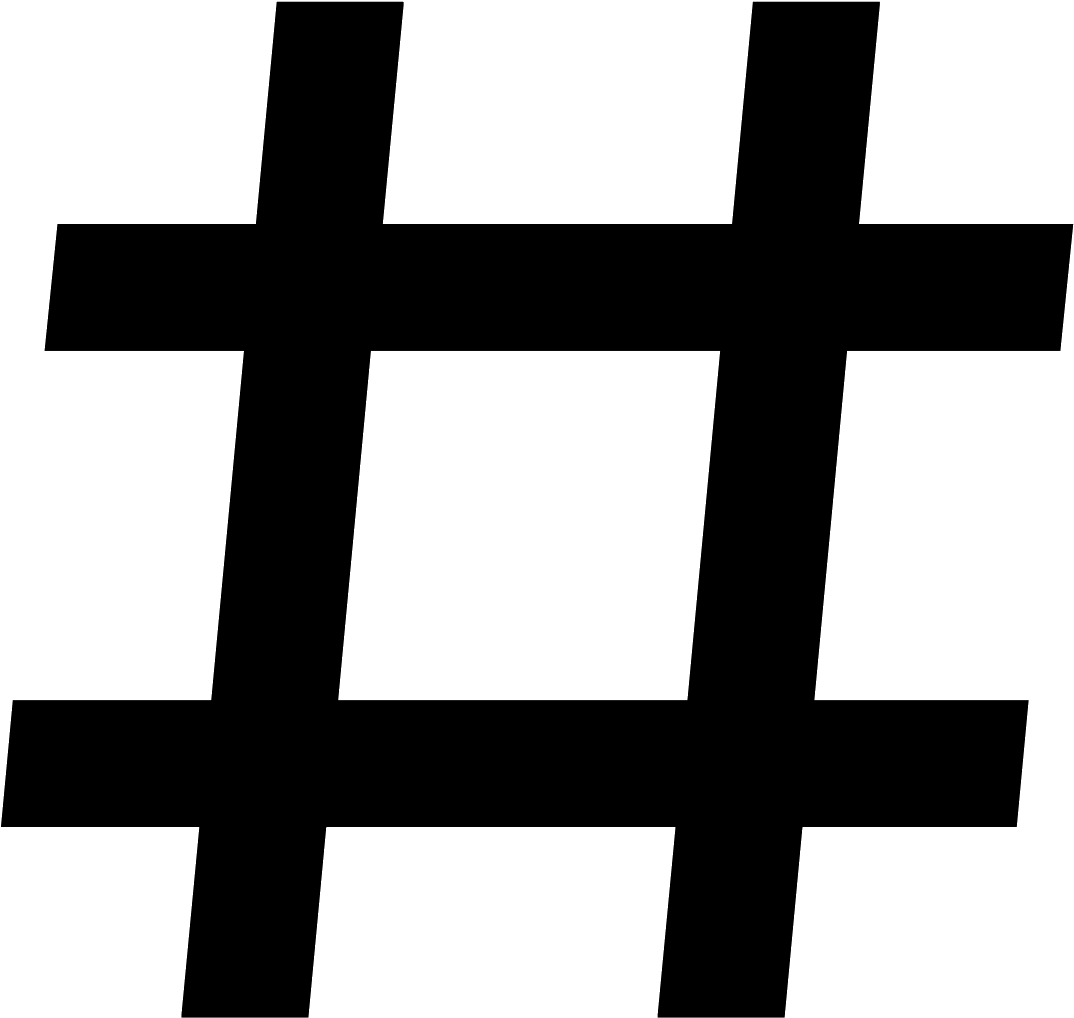}};%
        \draw (5,0) node[circle, draw, fill=white, minimum size=3pt, inner sep=1pt] (p5) {5};%
        \draw (5.5,0) node[circle, draw, fill=white, minimum size=3pt, inner sep=1pt] (p6) {6};%
        \draw (6,0) node[circle, draw, fill=white, minimum size=3pt, inner sep=1pt] (p7) {7};%
        \draw (7,0) node[circle, draw, fill=white, minimum size=3pt, inner sep=1pt] (p8) {8};%

        \node [below, yshift= -6pt] at (p1) {Provider};%
        \node [below, yshift= -6pt] at (cid) {CID};%
        \draw [decorate, decoration={brace, amplitude=5pt, raise=20pt}] (p6.east) -- (p3.west);%
        \node [below, yshift= -26pt] at ($(p6.east)!0.5!(p3.west)$) {Resolvers ($k$ closest peers)};%
        \node [above left] at (p1) {\notoemoji[height=0.5cm]{page_facing_up}};%

        \draw [->] (p1) to [out=60,in=120] (p3);%
        \draw [->] (p1) to [out=60,in=120] (p4);%
        \draw [->] (p1) to [out=60,in=120] (p5);%
        \draw [->] (p1) to [out=60,in=120] (p6);%
        \node (pr) [ellipse callout, draw, callout relative pointer={(0.25cm,-0.25cm)}] at (1,1.75) {1 has \textnotoemoji{page_facing_up}};%
        \node [anchor=west] at (pr.east) {Provider Record};%

        \node [anchor=west, inner sep=0] at (0,2.25) {\textbf{(a) Add Provider}};%

    \end{tikzpicture}%
    \hfill%
    \begin{tikzpicture}[]%
        \footnotesize%
        \draw[->] (0,0) -- (8,0) node[anchor=south] {keyspace};%
        \draw (0.5,0) node[circle, draw, fill=white, minimum size=3pt, inner sep=1pt] (p1) {1};%
        \draw (1.5,0) node[circle, draw, fill=white, minimum size=3pt, inner sep=1pt] (p2) {2};%
        \draw (2.5,0) node[circle, draw, fill=white, minimum size=3pt, inner sep=1pt] (p3) {3};%
        \draw (3.0,0) node[circle, draw, fill=white, minimum size=3pt, inner sep=1pt] (p4) {4};%
        \node at (4,0) (cid) {\includegraphics[width=0.5cm]{img/hashtag.png}};%
        \draw (5,0) node[circle, draw, fill=white, minimum size=3pt, inner sep=1pt] (p5) {5};%
        \draw (5.5,0) node[circle, draw, fill=white, minimum size=3pt, inner sep=1pt] (p6) {6};%
        \draw (6,0) node[circle, draw, fill=white, minimum size=3pt, inner sep=1pt] (p7) {7};%
        \draw (7,0) node[circle, draw, fill=white, minimum size=3pt, inner sep=1pt] (p8) {8};%

        \node [below, yshift= -6pt] at (p1) {Provider};%
        \node [below, yshift= -6pt] at (cid) {CID};%
        \draw [decorate, decoration={brace, amplitude=5pt, raise=20pt}] (p6.east) -- (p3.west);%
        \node [below, yshift= -26pt] at ($(p6.east)!0.5!(p3.west)$) {Resolvers ($k$ closest peers)};%
        \node [above left] at (p1) {\notoemoji[height=0.5cm]{page_facing_up}};%
        \node [below, yshift= -6pt] at (p8) {Downloader};%

        \draw [->] (p8) to [out=120,in=60] (p3);%
        \draw [->] (p8) to [out=120,in=60] (p4);%
        \draw [->] (p8) to [out=120,in=60] (p5);%
        \draw [->] (p8) to [out=120,in=60] (p6);%
        \node (pr) [ellipse callout, draw, callout relative pointer={(0.25cm,-0.25cm)}] at (3,1.75) {1 has \textnotoemoji{page_facing_up}};%
        \node [anchor=east] at (pr.west) {Provider Record};%
        \node (pr) [ellipse callout, draw, callout relative pointer={(-0.25cm,-0.25cm)}] at (6,1.75) {Who has \textnotoemoji{page_facing_up}?};%

        \node [anchor=west, inner sep=0] at (0,2.25) {\textbf{(b) Get Providers}};%
        
    \end{tikzpicture}%
    \hfill%
    \begin{tikzpicture}[]%
        \footnotesize%
        \draw[->] (0,0) -- (8,0) node[anchor=south] {keyspace};%
        \draw (0.5,0) node[circle, draw, fill=white, minimum size=3pt, inner sep=1pt] (p1) {1};%
        \draw (1.5,0) node[circle, draw, fill=white, minimum size=3pt, inner sep=1pt] (p2) {2};%
        \draw (2.5,0) node[circle, draw, fill=white, minimum size=3pt, inner sep=1pt] (p3) {3};%
        \draw (3.0,0) node[circle, draw, fill=white, minimum size=3pt, inner sep=1pt] (p4) {4};%
        \node at (4,0) (cid) {\includegraphics[width=0.5cm]{img/hashtag.png}};%
        \draw (5,0) node[circle, draw, fill=white, minimum size=3pt, inner sep=1pt] (p5) {5};%
        \draw (5.5,0) node[circle, draw, fill=white, minimum size=3pt, inner sep=1pt] (p6) {6};%
        \draw (6,0) node[circle, draw, fill=white, minimum size=3pt, inner sep=1pt] (p7) {7};%
        \draw (7,0) node[circle, draw, fill=white, minimum size=3pt, inner sep=1pt] (p8) {8};%

        \node [below, yshift= -6pt] at (p1) {Provider};%
        \node [above left] at (p1) {\notoemoji[height=0.5cm]{page_facing_up}};%
        \node [below, yshift= -6pt] at (p8) {Downloader};%

        \draw [->] (p1) to [out=30,in=150] (p8);%
        \node at (4,1.35) {\notoemoji[height=0.5cm]{page_facing_up}};%
        \node (pr) [ellipse callout, draw, callout relative pointer={(-0.25cm,-0.25cm)}, anchor=pointer] at (p8.north) {1 has \textnotoemoji{page_facing_up}};%

        \node [anchor=west, inner sep=0] at (0,2) {\textbf{(c) Content Transfer}};%
        
    \end{tikzpicture}%

    \caption{Content resolution using the DHT. 1) The provider upload its provider record to designated resolvers. 2) The downloader fetches the provider record from resolvers. 3) The downloader uses the information in the provider record to download the content directly from the provider. \vspaceaftercaption}
    \label{fig:add_get_provider}
\end{figure}

%% file: sections/model.tex
\vspacebeforesection
\section{Threat Model}\label{sec:model}

We assume $N$ DHT nodes participating in the IPFS network. Multiple nodes may share the same IP address (due to NAT or being hosted by the same physical machine)~\cite{marcus2018low}. However, two nodes cannot share the same ID.

We assume the presence of malicious actors in the network that may refuse to store valid provider records and distribute these to honest participants.
Malicious actors can spawn multiple virtual nodes within one physical machine, operate multiple physical machines, and coordinate their actions.
We assume that no honest node is fully eclipsed by malicious ones, \ie, each honest node has \emph{at least} one honest peer and the DHT routing allows honest nodes to reach any key and discover other honest peers.
IPFS already implements multiple mechanisms preventing eclipse attacks at the DHT level~\cite{total_eclipse}.

An attacker running Sybil nodes interfaces with the network using regular IPFS operations. We do not rely on bugs present in the operating system or any other components not related to the P2P network node implementation under attack. Any flaw in the P2P system protocols, however, may be exploited as these are considered part of the attack target. Non-Sybil DHT nodes (ones that are not spawned by the malicious actor) are assumed to be configured and operated as intended. Thus, importantly, the attacker only controls Sybil nodes that they create but does not corrupt or bring down any other nodes.

The goal of the attacker is to prevent \downloaders from obtaining \provider records for a target CID. This leads to \emph{content censorship} as the \downloader cannot find a \provider to obtain the content from (recall that \bitswap is only a cache for popular content while the DHT is required for reliable content discovery). We measure the attack effectiveness $\aeff$ as the ratio of unsuccessful $\Call{FindProviders}{\cid}$ queries to the total number of queries for existing content $\cid$ issued by honest \downloaders. A query is unsuccessful if it does not return any honest provider record. The attack effectiveness may vary depending on the placement of the CID and the \downloader issuing the request in the hash space. We thus always consider $\aeff$ as an average for multiple CIDs and multiple \downloaders, both uniformly spread across the hash space.

The goal of honest participants is to detect the attack and mitigate its effects, \ie to enable \downloaders to discover valid provider records and later fetch the content despite the actions of the attacker.

To assess the effectiveness of our \emph{detection} mechanism, we use the false positive $f_p$ and false negative $f_n$ rates. The false positive rate $f_p$ is the proportion of erroneous detections (\ie when there was no attack). The false negative rate $f_n$ is the proportion of attacks that are not detected. A detection leads to a mitigation action. A false positive leads, therefore, solely to additional overhead, while a false negative leads to effective censorship of content. Henceforth, we favor minimizing the false negative rate $f_n$.

The mitigation effectiveness $\meff$ is the ratio of the number of successful $\Call{FindProviders}{\cid}$ queries to the total number of queries issued by honest \downloaders when the target, existing $\cid$ is under attack and the mitigation mechanism is used.
A successful query is defined as one that returns at least one honest provider record. Similarly to the attack effectiveness, we report $\meff$ as an average for queries issued for multiple CIDs by multiple downloaders uniformly spread across the hash space. 

\begin{table}[t]
    \footnotesize
    \newcolumntype{E}{>{\raggedright\arraybackslash} m{0.125\linewidth} }
    \newcolumntype{F}{>{\raggedright\arraybackslash} m{0.775\linewidth} }
    \renewcommand{\arraystretch}{1.2}

    \begin{tabular}{EF}
    \toprule
    \multicolumn{2}{l}{\textbf{General parameters}} \\
    \midrule
    $N$ & Network size (number of nodes) \\
    $k$ & Bucket size, \resolvers per CID, and number of closest peers obtained in a $\Call{GetClosestPeers}{\cid}$ call (currently, $k=20$ in \texttt{libp2p}/IPFS) 
     \\
    \midrule
    
    \multicolumn{2}{l}{\textbf{Attack general parameters}} \\
    \midrule
    $\aeff$ &  Effectiveness of the attack $[\%]$ \\
    $\numSybils$ & Number of Sybil nodes \\
    $\numSybils(\aeff)$ & Number of Sybil nodes necessary to perform the attack with effectiveness $\aeff$. \\
    \midrule
    
    \multicolumn{2}{l}{\textbf{Attack costs}} \\
    \midrule
    $s(\numSybils)$ & Brute-force attempts necessary to generate $\numSybils$ Sybil identities that are the closest to a target CID.\\
    $\cgen$ & Cost of generating $s(\numSybils)$ Sybil identities $[\$]$\\
    $\coper$ & Cost per unit time of operating $\numSybils$ Sybil nodes to attack a single CID $[{\$}/{s}]$\\
    $\catt$ & Total cost of attacking a single CID $[\$]$\\
    \midrule
    
    \multicolumn{2}{l}{\textbf{Attack performance}} \\
    \midrule
    $\twarmup$ & Warmup time during which the $\numSybils$ Sybil nodes need to be run but the attack is not yet fully effective $[s]$\\
    $\teff$ &Time after $\twarmup$ during which the attack remains fully effective. The attack maintains its effect only as long as the $\numSybils$ Sybil nodes are present in the network $[s]$. \\
    \midrule
    
    \multicolumn{2}{l}{\textbf{Detection and mitigation performance}} \\
    \midrule
    $\threshold$ & Threshold for the detection mechanism \\
    $f_n$ & Detection false negative rate $[\%]$ \\
    $f_p$ & Detection false positive rate $[\%]$ \\
    $\meff$ & Effectiveness of the mitigation $[\%]$ \\
    \bottomrule
    
    \end{tabular}
    \caption{Parameters and characterization of the attack, detection, and mitigation.\vspaceaftercaption}
    \label{tab:notation}
\end{table}

%% file: sections/attack.tex
\vspacebeforesection
\section{CID Censorship Attack}\label{sec:attack}

We now proceed to detail our content censorship attack, targeting a specific victim CID.
To describe and analyze this attack, we use a Sybil attack model with notations adapted from prior work~\cite{prunster2018holistic, total_eclipse} and summarized in \Cref{tab:notation}.

\vspacebeforesection
\subsection{Attack process}
\label{sec:attack-process}

\input{img/attack-illustration.tex}

\para{Overview}
In IPFS, \provider records for a target CID should be stored on the $k=20$ peers closest to that CID (\ie \resolvers).
We generate $\numSybils \geq k$ Sybil identities that are closer to the target CID than to the closest honest peer. 
As a result, those Sybil identities receive all new \provider records from \providers and resolution queries from \downloaders for that CID.
\Cref{fig:attack-illustration} illustrates the position of Sybil peers.
Sybils drop \provider records and do not respond to queries for that target CID.

\para{Details}
The attack proceeds as follows.
First, we use the IPFS API to retrieve the current $k=20$ closest peer IDs to the target CID, sort them by distance to the CID, and identify the closest one.
We then repeatedly generate random public/private key pairs and compute new peer IDs by hashing the public key. 
If a generated peer ID is closer to the CID than to the currently closest peer ID, we keep the corresponding key pair, otherwise, we discard it. 
We repeat the process until we obtain $\numSybils$ peer IDs that are closer to the target CID than any of the original 20 closest peer IDs.

For each generated Sybil peer ID, we spawn a custom \texttt{libp2p} DHT node (server) that joins the IPFS network.
These Sybil nodes drop received provider records for the target CID and respond with empty messages to received resolution queries for the target CID.
The Sybil nodes behave normally for other CIDs, so that our experiments do not affect the rest of the network.
The attacker continuously monitors the set of $\numSybils$ closest peers to the target CID to make sure it contains only the Sybil nodes. 
If a new, honest node appears in the set, the attacker reacts by generating additional identifiers to maintain the desired number of Sybil nodes in the set.

\vspacebeforesection
\subsection{Attack analysis}
\label{sec:attack-analysis}

We analyze the attack in terms of cost and in terms of effectiveness, including its timing.

\para{Initial costs}
The first cost for an attacker is that of generating Sybil identities.
As the hash function is pre-image resistant, this process must use a brute force generation of private/public key pairs and associated IDs.
The number of attempts it takes for generating $\numSybils$ Sybil identities that are closer to the target is denoted $s(\numSybils)$.
This number naturally depends on $\numSybils$, but also on the distance of the closest honest peer from the target CID, which in turn depends on the number of peers in the network.
The closer the honest peer is to the target CID, the more keys the attacker needs to generate to obtain Sybil peer IDs closer to the CID.

The number of attempts further translates into an operational cost which we quantify using public cloud resource costs.
This cost $\cgen$ depends on $s(\numSybils)$ and the cost of generating one private/public key pair.
IPFS and \texttt{libp2p} support both the RSA and Edwards-curve Digital Signature Algorithm (EdDSA) cryptosystems.
The generation of keys for EdDSA is significantly faster than for RSA, and both are embarrassingly parallel; we evaluate these costs in \Cref{sec:evaluation}, and choose to target EdDSA in our implementation of the attack 

The effectiveness of the attack $\aeff$ varies with the number of Sybils $\numSybils$.
Theoretically, the attack only requires $\numSybils = k = 20$ Sybil identities.
However, different DHT nodes do not always discover the same set of $k=20$ closest peers (we provide experimental evidence in \Cref{sec:evaluation}, \Cref{fig:20closest}).
As a result, the attack generally requires more than $20$ Sybil identities (as some further honest peers may be discovered).
On the other hand, some of the discovered peers may not be online, so the attack may also succeed with fewer than $20$ Sybil identities.
We empirically studied the effect of the number of Sybil identities on the rate of successfully censoring the target content, and observe that $\numSybils = 45$ Sybil identities can censor content with a $\aeff=99\%$ probability of success (\Cref{sec:evaluation}, \Cref{fig:attack_success_rate}).
While $\numSybils(\aeff)$ depends on $k$, $\numSybils(\aeff)$ does not depend on 
whether the content was already provided before the Sybils were launched or not.

\para{Timing}
The \emph{warmup time} $\twarmup$, \ie, the time before the content is effectively censored and becomes undiscoverable, depends on whether the Sybil peers are launched \emph{before} the \provider sends its \provider records to the network or the Sybil peers are launched \emph{after}.
If Sybil peers are launched before the first \provider sends its \provider record to the network then, since the Sybil peers are the closest peers to the target CID, \providers will most likely send their \provider records to only Sybil peers.
These Sybil peers simply drop the provider records.
As a result, the content never becomes discoverable in the network.
Therefore, if all Sybil nodes are launched before the first \provider advertised the content, the attack is effective immediately, i.e., $\twarmup = 0$.
We note, however, that this best-case scenario is not likely in all contexts of use of IPFS, as it requires knowing the CID of the content to censor before mounting the attack.
This CID depends, indeed, on the \emph{content} of the file which may be known only upon its publication.
In certain cases, however, the attacker may know in advance that a specific file will be published and act to prevent its discoverability.

If the provider records were already stored on honest peers before the Sybil peers were launched, a \downloader may encounter an honest peer with relevant provider records \emph{before} reaching the Sybil peers, and be able to obtain the provider records this way.
By default, however, such provider records expire every 48 hours\footnote{24 hours in older versions of \texttt{go-libp2p}}, after which a \provider must call $\Call{Provide}{\cid}$ again.
As a result, in the worst case, the last provider record on an honest \resolver will be removed $\twarmup = 48$h after launching the Sybil nodes, after which the content becomes censored.
It is not desirable to get rid of this limited lifetime of provider records: an unlimited lifetime would result in a gradual overload of long-running peers and open new avenues for DoS attacks.

\para{Overall costs}
The overall costs of the attack $\catt$ include the initial costs $\cgen$ plus the operational costs of running the Sybil nodes at $\coper$ per unit time.
The operational cost is incurred during the warmup time $\twarmup$ before the attack is effective and the time during which the attack must \emph{remain} effective $\teff$.
Therefore, $\catt = \cgen  + (\twarmup + \teff) \times \coper$.

%% file: img/attack-illustration.tex
\begin{figure}[tb]%
    \centering%
    \begin{tikzpicture}[]%
        \footnotesize%
        \draw[->] (0,0) -- (8,0) node[anchor=south] {keyspace};%
        \draw (0.5,0) node[circle, draw, fill=white, minimum size=3pt, inner sep=1pt] (p1) {1};%
        \draw (1.5,0) node[circle, draw, fill=white, minimum size=3pt, inner sep=1pt] (p2) {2};%
        \draw (2.5,0) node[circle, draw, fill=white, minimum size=3pt, inner sep=1pt] (p3) {3};%
        \draw (3.0,0) node[circle, draw, fill=white, minimum size=3pt, inner sep=1pt] (p4) {4};%
        \draw (3.3,0) node[circle, draw, fill=red!50!white, minimum size=3pt, inner sep=1pt] (sybil1) {S};%
        \draw (3.6,0) node[circle, draw, fill=red!50!white, minimum size=3pt, inner sep=1pt] (sybil2) {S};%
        \node at (4,0) (cid) {\includegraphics[width=0.5cm]{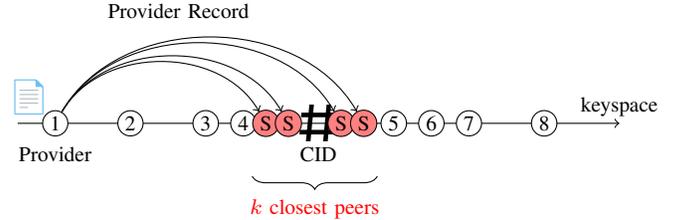}};%
        \draw (4.3,0) node[circle, draw, fill=red!50!white, minimum size=3pt, inner sep=1pt] (sybil3) {S};%
        \draw (4.6,0) node[circle, draw, fill=red!50!white, minimum size=3pt, inner sep=1pt] (sybil4) {S};%
        \draw (5,0) node[circle, draw, fill=white, minimum size=3pt, inner sep=1pt] (p5) {5};%
        \draw (5.5,0) node[circle, draw, fill=white, minimum size=3pt, inner sep=1pt] (p6) {6};%
        \draw (6,0) node[circle, draw, fill=white, minimum size=3pt, inner sep=1pt] (p7) {7};%
        \draw (7,0) node[circle, draw, fill=white, minimum size=3pt, inner sep=1pt] (p8) {8};%

        \node [below, yshift= -6pt] at (p1) {Provider};%
        \node [below, yshift= -6pt] at (cid) {CID};%
        \draw [decorate, decoration={brace, amplitude=5pt, raise=20pt}] (sybil4.east) -- (sybil1.west);%
        \node [below, yshift= -26pt] at ($(sybil4.east)!0.5!(sybil1.west)$) {\textcolor{red}{$k$ closest peers}};%
        \node [above left] at (p1) {\notoemoji[height=0.5cm]{page_facing_up}};%

        \draw [->] (p1) to [out=60,in=120] (sybil1);%
        \draw [->] (p1) to [out=60,in=120] (sybil2);%
        \draw [->] (p1) to [out=60,in=120] (sybil3);%
        \draw [->] (p1) to [out=60,in=120] (sybil4);%
        \node (pr) at (1,1.5) {};%
        \node [anchor=west] at (pr.east) {Provider Record};%

    \end{tikzpicture}%

    \caption{Illustration of a censorship attack: Attacker places $k$ Sybil peers ($k=4$ in this example) closer to the CID than any honest peers, so provider records are now only sent to Sybil peers. Similarly, requests to obtain provider records are sent to the Sybil peers, who can ignore them.\vspaceaftercaption}
    \label{fig:attack-illustration}
\end{figure}

%% file: sections/detection.tex
\vspacebeforesection
\section{Censorship Attack Detection}
\label{sec:detection}

The first step in countering an attack is its reliable detection. It enables activating mitigation techniques only when needed and avoids unnecessary overhead when the network is not under attack. Importantly, the detection cannot simply rely on the unavailability of the \provider records as 
this could be due to other reasons (e.g., the \provider records expired and the \provider did not renew them).
Moreover, the detection method should ideally expose
an attack as soon as the Sybil peers are added to the DHT, even though unavailability of the \provider records may only occur several hours ($\twarmup = 48$ hours) later
(as described in \Cref{sec:attack-analysis}). This would help prevent downtime for the content.
To execute the content censorship attack, the attacker must hijack all $\mathsf{PutProvider}$ advertisements for the target CID.
To this end, the attacker must operate \numSybils~Sybil peers whose IDs are closer to the target CID than to any other honest peer.
As a result, when an honest node queries the DHT for the target CID, the $20$ closest peer IDs it finds are closer to the CID than usual.
The node can thus use the observed peer IDs to detect whether the CID is censored or not.

\para{Method Overview}
We repurpose a statistical method originally developed by Cholez \emph{et al.}~\cite{cholez2010detection} for detecting eclipse attacks.
This method first obtains the $k=20$ closest peer IDs to the CID, using a DHT query, and computes the common prefix length of each peer ID with that CID (\ie the number of leading bits that match both the CID and the peer ID).
It then compares the empirical distribution of the $20$ common prefix lengths with the `model' probability distribution that would result if all these peers were honest, \ie, if their peer IDs were chosen randomly and uniformly.
We compare the empirical distribution to the model distribution by computing the Kullback-Liebler (KL) divergence between the two distributions (also known as a G-test~\cite{g-test}).
A large KL divergence indicates a mismatch between the empirical and model distributions, indicating an attack. On the other hand, a small KL divergence indicates a good match of the distributions, indicating no attack.
This KL divergence-based has been shown to perform well under a small number of samples~\cite{bio-statistics}, which is our situation with only $20$ samples corresponding to the $20$ closest peers.

A challenge of implementing this detection method for the IPFS network is that it requires knowing the number of peers in the DHT to compute the model distribution.
We demonstrate how to adapt this detection mechanism to a DHT of dynamic size by using a network size estimator.
Based on the network size estimate, our detection method computes an estimate of the model distribution, which is then used for detection.
Our detection method does not require any additional communication and requires little local computation.

\para{Method Details}
Honest peers choose their peer IDs uniformly at random, and thus under no attack, peer IDs in the DHT are distributed uniformly across the hash space.
We thus expect the common prefix lengths of peer IDs with the target CID to follow a geometric distribution. That is, given a target CID, the probability that the common prefix length of a randomly chosen peer ID with the target CID equals $x$ is $0.5^{x+1}$ (effects of the finite length ($256$ bits) of IDs can be neglected as the probability of a common prefix length of $256$ bits is negligible).
However, when a node queries the DHT for the target CID, it only obtains the $k=20$ closest peers to the target CID, and the distribution of their common prefix lengths is not geometric. Particularly, very small common prefixes are not observed as such peer IDs are far from the target ID (see~\Cref{fig:common-prefix-length-distribution}).
In the work of Cholez \emph{et al.}~\cite{cholez2010detection}, this effect is modeled by using a geometric distribution conditioned on the common prefix length being greater than $L$, where $L$ is a parameter that is estimated using measurements from the network.
If we adopted this approach, we would have to compute $L$ from a fixed, pre-determined network size estimate. We found that the model distribution computed using this method is highly susceptible to small variations in the network size estimate.
Instead, in this work, we compute the model distribution dynamically, \ie, as it varies according to an estimation of the network size.

\begin{figure}[t]
    \centering
    \includegraphics[width=\columnwidth]{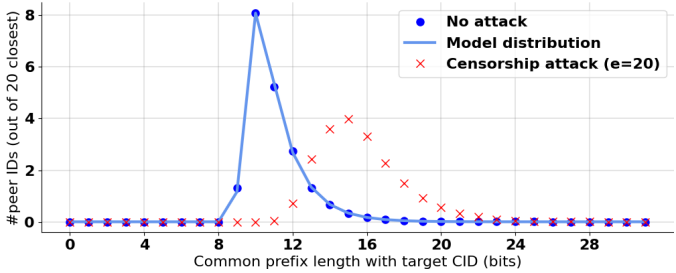}
    \caption{Probability distribution of common prefix lengths of the target CID with its $k=20$ closest peer IDs.\vspaceaftercaption}
    \label{fig:common-prefix-length-distribution}
\end{figure}

We will first describe the method assuming that the network size $N$ is known. 
Then, we will show how to adapt this method using an estimation of the network size.
We recall that the common prefix length of a randomly chosen peer ID with the target CID is distributed geometrically.
Now, in response to a lookup for the target CID, a node obtains the $k$ largest common prefix lengths (corresponding to the $k$ smallest distances). The cumulative distribution function of the $j$-th largest prefix length can be calculated as
\begin{align}
    \label{equ:cdf_j_best}
	F_j(x) &\triangleq \mathrm{Pr}(X^{(j)} \leq x) \nonumber \\
    &= \begin{cases} \sum_{i=0}^{j-1} {N \choose i} \left( 1-0.5^{x+1} \right)^{N-i} 0.5^{(x+1)i} & x \geq 0, \\ 0 & x < 0. \end{cases}
\end{align}
Since we only receive one sample for the $j$-th best peer ID for each $j$, we can compute an average probability mass function of the $k$ best common prefix lengths as
\begin{equation}
    \label{equ:pmf_avg}
    p(x) = \frac{1}{k} \sum_{j=1}^{k} (F_j(x) - F_j(x-1)).
\end{equation}
This model distribution $p(x)$ is shown along with the empirical average distribution of common prefix lengths in \Cref{fig:common-prefix-length-distribution}. Note that the empirical distribution in case of no attack (all peer IDs generated randomly) exactly matches the model distribution, while the distribution under an attack ($\numSybils = 20$ Sybil peers closer than honest peer IDs) significantly differs from the model.
Importantly, placing more than $20$ Sybils drives the ID distribution further away from the model, easing the detection.
Since $k=20$, the sums in \cref{equ:cdf_j_best,equ:pmf_avg} can be computed efficiently.

The KL divergence is a tool used to quantify the difference between two probability distributions. For two discrete probability distributions $p(x)$ and $q(x)$ on a support set $\mathcal{X}$, the KL divergence from $p$ to $q$ is defined as
\begin{equation}
    \label{equ:kl-divergence-definition}
    D\left(q \mathbin\Vert p \right) \triangleq \sum_{x \in \mathcal{X}} q(x) \ln\left( \frac{q(x)}{p(x)}\right).
\end{equation}
In our case, $X$ is the random variable denoting the common prefix length between a peer ID and the target CID,
$\mathcal{X} = \{0,...,256\}$, 
$p(x)$ is the model distribution and  
\begin{equation}
	\label{equ:pmf_emp}
    q(x) \triangleq \frac{1}{k} \sum_{i=1}^{k} \mathbbm{1}\{X_i = x\}
\end{equation} is the empirical distribution of common prefix lengths. 
Note that the support of the empirical distribution $q(x)$ is a subset of the support of the model distribution $p(x)$, therefore the sum in the KL divergence is only computed over values of $x$ for which $q(x) > 0$.
Now, a threshold $\threshold$ must be chosen so that the CID is flagged to be under an attack if and only if $D\left(q \mathbin\Vert p \right) > \threshold$.

\para{Adapting to Dynamic Network Sizes}
Computing the model distribution requires knowing the number of peers in the DHT ($N$) that is unknown to the DHT nodes.
Instead, we substitute an estimate $\hat{N}$ of the network size obtained from the network size estimator we described in \Cref{sec:netsize}.
A node can locally estimate the model distribution by substituting $\hat{N}$ instead of $N$ in \cref{equ:cdf_j_best,equ:pmf_avg}.
The complete censorship detection algorithm is specified in \Cref{alg:eclipse-detection-algorithm}.

It is important to note that the network size estimate does not depend on the closest peer IDs for the CID for which censorship detection is being performed. Instead, it is computed using the closest peer IDs to randomly chosen keys. Therefore, recall from \Cref{sec:netsize} that the attacker would require a number of Sybil peers of the order of the network size to bias the estimator. As a result, the detection remains robust under the censorship attack. 

\para{Choosing a detection threshold}
The detection threshold $\threshold$ is a per-node constant value.
A higher threshold results in more false negatives (some attacks go undetected, hence unmitigated) while a lower threshold results in more false positives (mitigation overhead when there is no attack).
In \Cref{sec:evaluation}, we evaluate the false positive and false negative rates for different thresholds and recommend a threshold that favors reducing false negatives.
Different nodes can however choose different thresholds according to their desired trade-off between the error rates.
In \Cref{sec:evaluation}, we show that a constant threshold suffices even as the network size varies.

\begin{algorithm}[t]%
    \caption{%
        Censorship Detection Algorithm; $\threshold$ is a pre-decided detection threshold
    }%
    \label{alg:eclipse-detection-algorithm}%
    \begin{algorithmic}[1]%
        \Procedure{Detection}{$\key$}
        \State $\algvar{peers} \gets \Call{Get20ClosestPeers}{\key}$
        \State $q \gets \operatorname{numPeersPerCPL}(\algvar{peers})/20$ \Comment{\cref{equ:pmf_emp}}
        \State $N \gets \Call{GetNetsizeEstimate}{ }$ \Comment{\cref{equ:avg-dist,equ:netsize-least-squares}}
        \State $p \gets \operatorname{computeModelDist}(N)$ \Comment{\cref{equ:cdf_j_best,equ:pmf_avg}}
        \State $\algvar{KL} \gets \operatorname{computeKL}(p, q)$ \Comment{\cref{equ:kl-divergence-definition}}
        \State \Return $\algvar{KL} > \threshold$ \Comment{true indicates attack}
        \EndProcedure
    \end{algorithmic}%
\end{algorithm}%

%% file: sections/mitigation.tex
\vspacebeforesection
\section{Mitigation with Region-Based Queries}
\label{sec:mitigation}

Countering Sybil attacks in an open, decentralized system is challenging.
Traditionally, this problem is solved by binding identities to valuable resources (\eg using Proof of Work~\cite{dwork92proofofwork,baumgart2007s}), certificate authorities~\cite{castro_dht}, reputation systems~\cite{sybillimit, sybilguard, danezis2005sybil, whanau}, or diversifying the IP addresses of the peers of each node~\cite{total_eclipse}.
Proof of Work and IP address restrictions are not sufficient as our attack only requires a few Sybil peers ($e \approx 45$): these measures would only slightly increase the cost of the attack. On the other hand, certificate authorities and reputation systems hamper the decentralization and open participation model of IPFS.
Simply increasing the value of the number of closest peers contacted in a DHT query (currently $k=20$) also does not solve the problem as the attacker only needs to generate more Sybil peers to match the new number.
Another naive idea is that the providers modify the content by one bit to modify its CID. However, the new CID must be then advertised to potential downloaders to make the content publicly accessible. The attacker can then simply “follow” the new CIDs and continuously censor the content.
Further, modifying the content is unsuitable for immutable Web3.0 content (e.g., NFTs and DIDs) whose hash is already published on a blockchain.

The fundamental problem in countering the content censorship attack lies in the inability to classify \resolvers as honest or malicious. When a \downloader receives no \provider record or an inactive provider record from a \resolver (\ie, it is unable to find the referenced provider or the provider does not hold the content), this can be due to several reasons. For instance, the \resolver was offline, the record used to be correct but the \provider had since left, or there was a network failure.
As a result, the \downloader cannot draw any conclusions on the \resolver's correctness based on the received results, eliminating any attempts to gradually filter out malicious nodes by local scoring systems.

\input{img/mitigation-illustration.tex}

\para{Main idea}
The core observation behind our approach is that, while an attacker can spawn additional Sybil identities, it has no way of removing the honest ones from the network. As long as the \provider can send its provider record to the initial honest \resolvers, and the \downloader can communicate with these \resolvers, the censorship attack will be mitigated. To maintain communication with the initial honest \resolvers even during an attack, we propose \emph{region-based} DHT queries. Rather than communicating with the $k=20$ closest peers to a CID, that an attacker can easily control, we communicate with all the nodes in the hash space region that $k=20$ uniformly distributed peer IDs should cover (\Cref{fig:mitigation}).
The size of this region is calculated using the network size estimate and using the assumption that honest peer IDs are distributed uniformly over the key space.
Such an approach ensures that regardless of the number of Sybil nodes placed by an attacker, the \provider can store provider records on $\approx 20$ honest \resolvers, and the \downloader also communicates with $\approx 20$ honest \resolvers to reliably retrieve the correct provider records.
To prevent additional overhead when there is no attack, we run the region-based queries only when an attack is detected using the detection mechanism that we detailed in \Cref{sec:detection}.

\begin{algorithm}[t]
    \caption{Function to find all peers with a Common Prefix Length (CPL) $\geq \algvar{minCPL}$ with $\algvar{key}$} 
    \label{alg:region_queries}
    \begin{algorithmic}[1]
    \Procedure{FindByCPL}{$\algvar{key}, \algvar{minCPL}$}
        \State $\algvar{set} \gets \Call{GetClosestPeers}{\algvar{key}}$
        \State $\algvar{CPL} \gets \operatorname{minCommonPrefixLength}(\algvar{set}, \algvar{key})$
        \While{$\algvar{CPL} \geq \algvar{minCPL}$}
            \State $\algvar{qkey} \gets \algvar{key}[\mathbin{:} \algvar{CPL}] \mathbin\Vert \overline{\algvar{key}[\algvar{CPL}]} \mathbin\Vert \algvar{key}[\algvar{CPL}+1 \mathbin{:}]$
            \State $\algvar{set} \gets \algvar{set} \cup \Call{FindByCPL}{\algvar{qkey}, \algvar{CPL}+1}$
            \State $\algvar{CPL} \gets \algvar{CPL} - 1$
        \EndWhile
        \State $\operatorname{removeItemsWithPrefixLessThan}(\algvar{set}, \algvar{minCPL})$
        \State \Return $\algvar{set}$
        \EndProcedure
        \end{algorithmic} 
\end{algorithm}

\para{Algorithm}
The region-based query algorithm is described in \Cref{alg:region_queries}, with a sample execution in~\Cref{fig:region-based-illustration}.
The goal of this algorithm is to find all peer IDs that share a common prefix of at least $\algvar{minCPL}$ bits with $\key$.
Note that any two keys $\algvar{k}_1, \algvar{k}_2$ have a common prefix length (CPL) of at least $l$ iff the XOR distance between $\algvar{k}_1$ and $\algvar{k}_2$ is less than $2^{256-l}$.
Therefore, the common prefix requirement specifies a region of the key space with a distance $2^{256-\algvar{minCPL}}$ from $\key$.
To keep our mitigation compatible with the current version of \texttt{libp2p} DHT nodes, we use the same RPCs that the DHT nodes currently use.
Therefore, we build the algorithm using only calls to $\Call{GetClosestPeers}{\key}$ which obtains the 20 peer IDs that are the closest to $\key$, which is already available in \texttt{go-libp2p-kad-dht}~\cite{libp2p_github_get_closest_peers}.  We start with this primitive and then compute the common prefix length shared by $\key$ and all of its 20 closest peer IDs, which we note as $\CPL$.

Since we have found at least one peer ID with a common prefix length $\CPL$, we must have found all peer IDs with common prefix length $\geq \CPL+1$, as the latter are closer (in XOR distance) to $\key$ than the former (see step 0 in \Cref{fig:region-based-illustration}).
In the next step, we would like to find all peer IDs with common prefix $\geq \CPL$ with $\key$. 
Since we have already found all peer IDs with the prefix $\key[\mathbin{:} \CPL+1]$ (\ie, the first $\CPL+1$ bits match $\key$),
we only need to find all peers IDs with the prefix $\key[\mathbin{:} \CPL] \mathbin\Vert \overline{\key[\CPL]}$ (\ie the first $\CPL$ bits match $\key$ and the $(\CPL+1)$-th bit is different).
This is done recursively using our algorithm.
This step is repeated until all peer IDs with common prefix length $\geq \algvar{minCPL}$ with $\key$ have been found.

\input{img/region-based-illustration.tex}

While using this region-based query algorithm, we choose the value of $\algvar{minCPL}$ such that a region of the key space with common prefix length at least $\algvar{minCPL}$ with $\key=\cid$ contains at least $k=20$ honest peer IDs with high probability.
Suppose that there are a total of $N$ peer IDs in the DHT, distributed uniformly across the hash space. A common prefix length of at least $\algvar{minCPL}$ corresponds to a XOR distance $<2^{256 - \algvar{minCPL}}$. Then, the expected number of peer IDs in this region is $2^{\algvar{minCPL}} \times N$.
By setting $\algvar{minCPL} = \lceil \log_2\left(\frac{N}{k}\right) \rceil$, we have a region that contains $k$ honest peer IDs on average.
Given an estimate $\hat{N}$ of the network size (as described in \cref{sec:netsize}), we substitute $\hat{N}$ for $N$ to calculate the region size.
By a simple probabilistic bound, we can also extend this region to contain $k$ honest peer IDs with high probability.

\para{Cost Analysis}
By default, both \providers and \downloaders do one DHT lookup (using $\Call{GetClosestPeers}{\key}$) to obtain the list of $k=20$ closest peer IDs to $\key$. 
When a \provider or \downloader uses a region-based query, it does multiple lookups using $\Call{GetClosestPeers}{\cdot}$. 
The number of lookups required increases sub-linearly in the number of Sybil identities placed by an attacker (shown experimentally in \Cref{fig:mit-lookups-sybils}).
Importantly, operating a Sybil identity requires participating in the DHT routing and responding to keep-alive messages.
As a result, the cost for the attacker increases linearly with the number of Sybil identities.
When the target CID is not under attack, using the region-based query would still use more than one $\Call{GetClosestPeers}{\cdot}$ lookups, because honest peer IDs are distributed randomly, and therefore the chosen region might contain more than $20$ peer IDs.
To avoid this overhead when there is no attack, we run the region-based lookup only when the detection mechanism (\Cref{sec:detection}) detects an attack, and use the default lookup otherwise.
We evaluate the \provider's and \downloader's cost of the region-based queries (number of lookups and latency) and the attacker's cost in \Cref{sec:evaluation}.

\para{Correctness Analysis}
We prove that \Cref{alg:region_queries} indeed finds all peer IDs with a common prefix length of at least $\algvar{minCPL}$ with $\key$.
\begin{theorem}
    \label{thm:region-based-routing-proof}
    Assuming that $\Call{GetClosestPeers}{\key}$ returns the $20$ closest peer IDs to $\key$, $\Call{FindByCPL}{\algvar{key}, \algvar{minCPL}}$ returns all peer IDs with a common prefix of at least $\algvar{minCPL}$ bits with $\key$.
\end{theorem}
\begin{proof}
    We prove this claim through induction.
    For the base case, if there are $< 20$ peer IDs with a common prefix length of at least $\algvar{minCPL}$ with $\key$, then $\Call{GetClosestPeers}{\key}$ must return at least one peer with a common prefix length $ < \algvar{minCPL}$ with $\key$. Therefore, $\CPL < \algvar{minCPL}$, hence the function returns all peer IDs that have a common prefix length of at least $\algvar{minCPL}$.

    Otherwise, $\CPL \geq \algvar{minCPL}$. Since we have found at least one peer with a common prefix length $\CPL$, we have found all peers with a common prefix length $\geq \CPL + 1$, that is all peers with the prefix $\key[:\CPL+1]$. Thus, we create a new key $\algvar{qkey}$ which has the prefix $\key[:\CPL] \mathbin\Vert \overline{\key[\CPL]}$. By induction, we assume that $\Call{findByCPL}{\algvar{qkey}, \CPL+1}$ returns all peers with prefix $\key[\mathbin{:} \CPL] \mathbin\Vert \overline{\key[\CPL]}$. Together, we now have all peers with the prefix $\key[\mathbin{:} \CPL]$. After subtracting $1$ from $\CPL$, we maintain the invariant that we have found all peers with prefix $\key[\mathbin{:}\CPL+1]$. If the loop doesn’t quit, then we continue to find peers with one more bit in the common prefix in every iteration. If the loop quits, this means that $\CPL + 1 \leq \algvar{minCPL}$, therefore we have found all peers with common prefix length of at least $\algvar{minCPL}$ as promised. 
\end{proof}

Even if $\Call{GetClosestPeers}{\key}$ does not return \textit{all} of the 20 closest peer IDs to $\key$, the algorithm will still terminate (because $\CPL$ decreases at every iteration) but may not find all the peers with a common prefix length of at least $\algvar{minCPL}$.
Since we restrict ourselves to build the region-based lookup using $\Call{GetClosestPeers}{\key}$, our method is accurate only in the cases when $\Call{GetClosestPeers}{\key}$ is accurate.

%% file: img/mitigation-illustration.tex
\begin{figure}[h]
    \centering
    \begin{tikzpicture}[]%
        \footnotesize%
        \draw[->] (0,0) -- (8,0) node[anchor=south] {keyspace};%
        \draw (0.5,0) node[circle, draw, fill=white, minimum size=3pt, inner sep=1pt] (p1) {1};%
        \draw (1.5,0) node[circle, draw, fill=white, minimum size=3pt, inner sep=1pt] (p2) {2};%
        \draw (2.5,0) node[circle, draw, fill=white, minimum size=3pt, inner sep=1pt] (p3) {3};%
        \draw (3.0,0) node[circle, draw, fill=white, minimum size=3pt, inner sep=1pt] (p4) {4};%
        \draw (3.3,0) node[circle, draw, fill=red!50!white, minimum size=3pt, inner sep=1pt] (sybil1) {S};%
        \draw (3.6,0) node[circle, draw, fill=red!50!white, minimum size=3pt, inner sep=1pt] (sybil2) {S};%
        \node at (4,0) (cid) {\includegraphics[width=0.5cm]{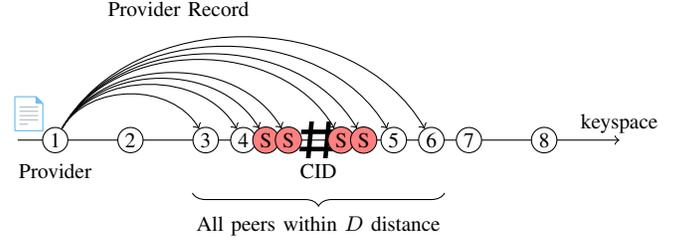}};%
        \draw (4.3,0) node[circle, draw, fill=red!50!white, minimum size=3pt, inner sep=1pt] (sybil3) {S};%
        \draw (4.6,0) node[circle, draw, fill=red!50!white, minimum size=3pt, inner sep=1pt] (sybil4) {S};%
        \draw (5,0) node[circle, draw, fill=white, minimum size=3pt, inner sep=1pt] (p5) {5};%
        \draw (5.5,0) node[circle, draw, fill=white, minimum size=3pt, inner sep=1pt] (p6) {6};%
        \draw (6,0) node[circle, draw, fill=white, minimum size=3pt, inner sep=1pt] (p7) {7};%
        \draw (7,0) node[circle, draw, fill=white, minimum size=3pt, inner sep=1pt] (p8) {8};%

        \node [below, yshift= -6pt] at (p1) {Provider};%
        \node [below, yshift= -6pt] at (cid) {CID};%
        \draw [decorate, decoration={brace, amplitude=5pt, raise=20pt}] (p6.east) -- (p3.west);%
        \node [below, yshift= -26pt] at ($(p6.east)!0.5!(p3.west)$) {All peers within $D$ distance};%
        \node [above left] at (p1) {\notoemoji[height=0.5cm]{page_facing_up}};%

        \draw [->] (p1) to [out=60,in=120] (sybil1);%
        \draw [->] (p1) to [out=60,in=120] (sybil2);%
        \draw [->] (p1) to [out=60,in=120] (sybil3);%
        \draw [->] (p1) to [out=60,in=120] (sybil4);%
        \draw [->] (p1) to [out=60,in=120] (p3);%
        \draw [->] (p1) to [out=60,in=120] (p4);%
        \draw [->] (p1) to [out=60,in=120] (p5);%
        \draw [->] (p1) to [out=60,in=120] (p6);%
        \node (pr) at (1,1.75) {};%
        \node [anchor=west] at (pr.east) {Provider Record};%
    \end{tikzpicture}

    \caption{Illustration of our mitigation: A provider record is sent to all peers within a region that contains $k$ honest peers on average ($k=4$ in this example), as estimated using the network size. This region contains $k$ honest peers even in the presence of Sybil peers, as honest peers are not removed.\vspaceaftercaption \vspaceaftercaption}
    \label{fig:mitigation}
\end{figure}

%% file: img/region-based-illustration.tex
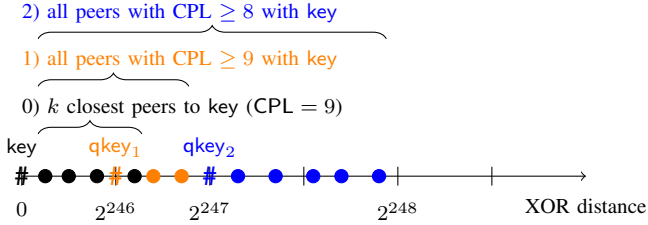
\begin{figure}[tb]%
    \centering%
    \begin{tikzpicture}[scale=1.25]%
        \footnotesize%
        \draw[->] (0,0) -- (6,0) node[anchor=north, yshift=-0.2cm] {XOR distance};%
        \foreach \x in {0,1,2,3,4,5} {%
            \draw (\x,0) ++(0,0.1) -- ++(0,-0.2);%
        }%
        \node [below] at (0,-0.2) {$0$};%
        \node [below] at (1,-0.2) {$2^{246}$};%
        \node [below] at (2,-0.2) {$2^{247}$};%
        \node [below] at (4,-0.2) {$2^{248}$};%
        \node [minimum size=3pt] at (0,0) (cid) {{\normalsize\textbf{\#}}};%
        \node [above, yshift=0.15cm] at (0,0) {$\algvar{key}$};%

        \draw (0.25,0) node[circle, draw, fill=black, minimum size=5pt, inner sep=0] (p1) {};%
        \draw (0.5,0) node[circle, draw, fill=black, minimum size=5pt, inner sep=0] (p2) {};%
        \draw (0.8,0) node[circle, draw, fill=black, minimum size=5pt, inner sep=0] (p3) {};%
        \draw (1.2,0) node[circle, draw, fill=black, minimum size=5pt, inner sep=0] (p4) {};%

        \draw [decorate, decoration={brace, amplitude=5pt, raise=16pt}] (p1.west) -- (p4.east);%
        \node [anchor=south west, yshift= 22pt, inner sep=0] at (0,0) {0) $k$ closest peers to $\key$ ($\CPL = 9$)};%

        \def\colorone{orange}%
        \node [minimum size=3pt] at (1,0) (qkey1) {{\normalsize\textcolor{\colorone}{\textbf{\#}}}};%
        \node [above, yshift=0.15cm] at (qkey1) {\textcolor{\colorone}{$\algvar{qkey}_1$}};%

        \draw (1.4,0) node[circle, draw=\colorone, fill=\colorone, minimum size=5pt, inner sep=0] (p5) {};%
        \draw (1.7,0) node[circle, draw=\colorone, fill=\colorone, minimum size=5pt, inner sep=0] (p6) {};%

        \draw [decorate, decoration={brace, amplitude=5pt, raise=34pt}] (p1.west) -- (p6.east);%
        \node [anchor=south west, yshift= 40pt, inner sep=0] at (0,0) {\textcolor{\colorone}{1) all peers with CPL $\geq 9$ with $\key$}};%

        \def\colortwo{blue}%
        \node [minimum size=3pt] at (2,0) (qkey2) {{\normalsize\textcolor{\colortwo}{\textbf{\#}}}};%
        \node [above, yshift=0.15cm] at (qkey2) {\textcolor{\colortwo}{$\algvar{qkey}_2$}};%

        \draw (2.3,0) node[circle, draw=\colortwo, fill=\colortwo, minimum size=5pt, inner sep=0] (p7) {};%
        \draw (2.7,0) node[circle, draw=\colortwo, fill=\colortwo, minimum size=5pt, inner sep=0] (p8) {};%
        \draw (3.1,0) node[circle, draw=\colortwo, fill=\colortwo, minimum size=5pt, inner sep=0] (p9) {};%
        \draw (3.4,0) node[circle, draw=\colortwo, fill=\colortwo, minimum size=5pt, inner sep=0] (p10) {};%
        \draw (3.8,0) node[circle, draw=\colortwo, fill=\colortwo, minimum size=5pt, inner sep=0] (p11) {};%

        \draw [decorate, decoration={brace, amplitude=5pt, raise=52pt}] (p1.west) -- (p11.east);%
        \node [anchor=south west, yshift= 58pt, inner sep=0] at (0,0) {\textcolor{\colortwo}{2) all peers with CPL $\geq 8$ with $\key$}};%
        
    \end{tikzpicture}%
    \caption{%
        An example illustration of $\textsc{FindByCPL}(\algvar{key}, 8)$ 
        (see \cref{alg:region_queries}). Note that any two keys $\algvar{k}_1, \algvar{k}_2$ have common prefix length (CPL) at least $l$ iff the XOR distance between $\algvar{k}_1$ and $\algvar{k}_2$ is less than $2^{256-l}$. Step 0: Find the $k$ ($4$ in this example) closest peers to $\key$ (shown in black). Calculate their minimum common prefix length with $\key$ ($\CPL = 9$ in this example). Step 1: $\textsc{FindByCPL}(\algvar{qkey}_1, 10)$ returns the green peers with CPL $\geq 10$ with $\algvar{qkey}_1$. Now we have found all peers with CPL $\geq 9$ with $\key$. Step 2: $\textsc{FindByCPL}(\algvar{qkey}_2, 9)$ returns the blue peers. Now we have found all peers with CPL $\geq 8$ with $\key$, and the algorithm terminates.\vspaceaftercaption \vspaceaftercaption
    }%
    \label{fig:region-based-illustration}%
\end{figure}%

%% file: sections/evaluation.tex
\vspacebeforesection
\section{Evaluation}
\label{sec:evaluation}

In this section, we evaluate the cost of the attack and the effectiveness of our detection and mitigation methods.

\vspacebeforesection
\subsection{Setup}
\label{sec:evaluation_setup}

We perform our experiments on the live IPFS networks. All our attacks targeted only content that we created. No content provided by other users was affected. Although we did not require any special permissions from Protocol Labs to execute our attacks, we informed them for responsible disclosure.

In the evaluation, we use three types of DHT nodes: \textit{(i)} the malicious Sybil nodes, \textit{(ii)} a \provider node hosting the content that we created, and \textit{(iii)} a \downloader node that attempts to resolve the target CID and fetch its content. 
All these nodes were hosted on a single AWS $\mathsf{t3.xlarge}$ instance with $4$ vCPUs and $16$ GiB memory.
The attacker node is implemented as a custom DHT client using \texttt{libp2p}~\cite{libp2p_github}. Our mitigation and detection methods are also implemented on top of \texttt{libp2p}. 
Except for the experiment of \Cref{fig:eclipse_timeline}, the \provider sends the provider record (\ie, it provides the content) after the Sybils are launched so that we avoid waiting for $\twarmup=48$ hours for the attack to take effect.

\vspacebeforesection
\subsection{DHT Lookup Accuracy}

\begin{figure}[tb]%
    \centering%
    \includegraphics[width=\linewidth]{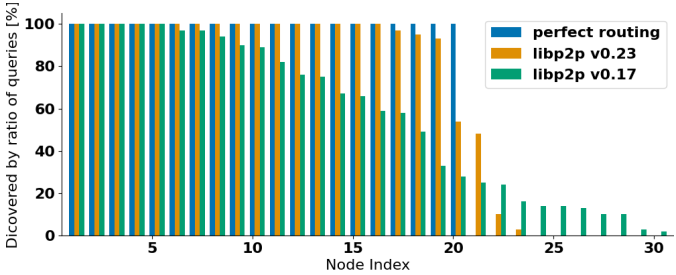}%
    \caption{%
        Distribution of nodes discovered during a DHT query by different versions of IPFS clients, across 100 experiments. \vspaceaftercaption \vspaceaftercaption%
    }%
    \label{fig:20closest}%
\end{figure}%

The DHT query $\Call{GetClosestPeers}{\key}$ is ideally expected to return the $k=20$ closest peer IDs to the queried $\key$.
\Cref{fig:20closest} shows the number of $\Call{GetClosestPeers}{\key}$ queries in which each peer close to the queried $\key$ is discovered. The x-axis shows the index of nodes ordered by their distance to the queried $\key$ ($1$ is the closest node). The y-axis shows the fraction of queries in which each node was discovered. The results are averaged over $100$ experiments. The set of peers obtained is compared with the true $20$ closest peers (`perfect routing') that were obtained from a network crawler. We observed that when different nodes perform this query, and when the same node does it multiple times, the set of $20$ peers received in response is not always the same. This effect is present in both versions of \texttt{libp2p} considered here but the query responses are more consistent in the newer version. To mitigate this effect, the attacker has to use a larger number of Sybils $e > 20$ for high attack effectiveness $\aeff \to 100\%$.

\vspacebeforesection
\subsection{Attack}

\begin{figure}[tb]%
    \centering%
    \includegraphics[width=\linewidth]{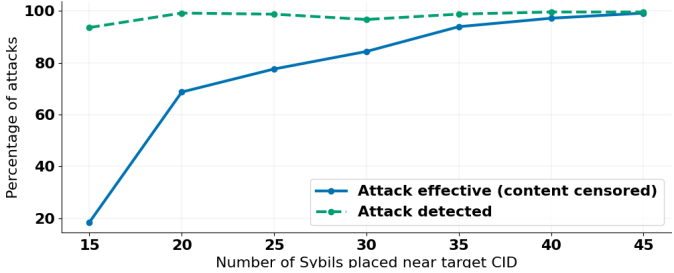}%
    \caption{%
        Percentage of attacks in which \textit{i)} the attack was effective and \textit{ii)} the detection algorithm detected the attack. \vspaceaftercaption%
    }%
    \label{fig:attack_success_rate}%
\end{figure}%

We follow by determining the number of Sybil identities $\numSybils$ required to achieve high attack effectiveness $\aeff$.
Since the DHT queries do not consistently return the same set of peers, \downloaders may discover some honest \resolvers when there are only $\numSybils = 20$ Sybils.
\Cref{fig:attack_success_rate} shows the success
rate of the attack as the number of Sybil peers varies.
Even when $\numSybils < 20$, the attack sometime succeeds because too few honest \resolvers are contacted, and they may be offline.
Involving more Sybil peers increases the chance of a successful attack but also
proportionally increases the cost for the attacker.
With $e=45$ Sybils, the attack succeeds with $99\%$ probability. We use this value for the rest of the experiments.

 \begin{figure}[tb]%
    \centering%
    \includegraphics[width=\linewidth]{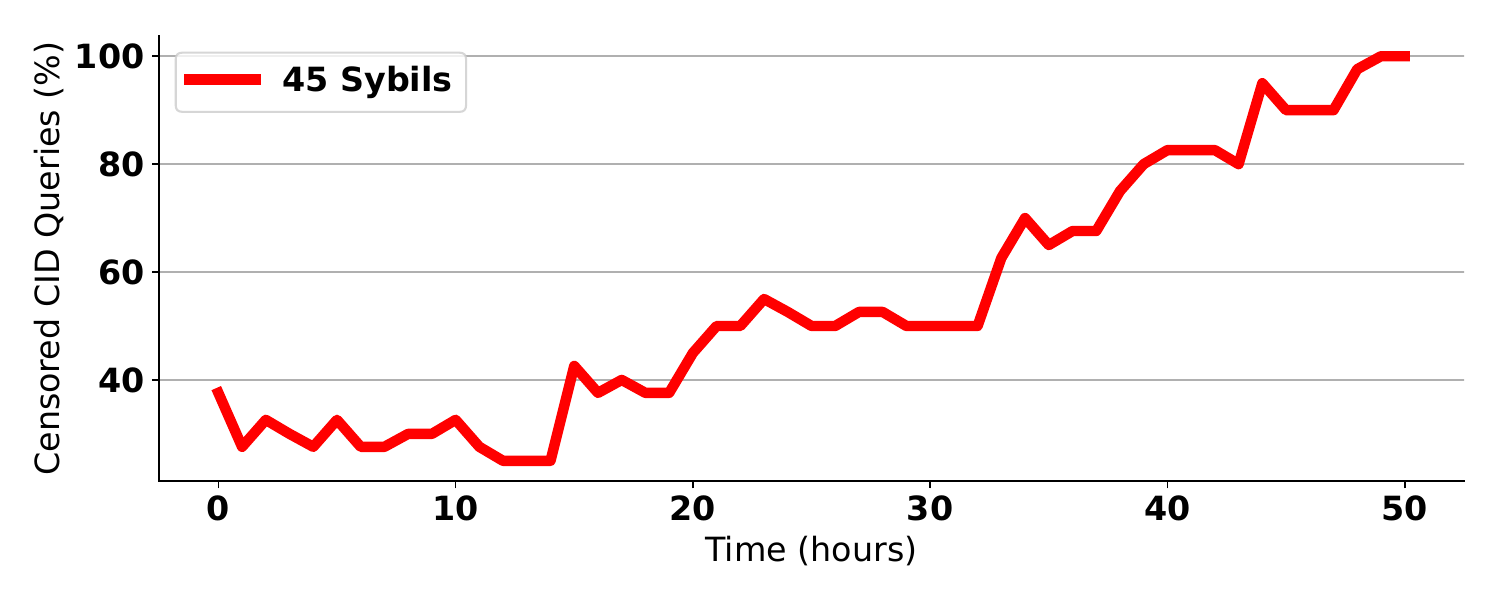}%
    \caption{%
        Rate of censored CID queries (\%) over time, from the start of launching the attack (when content has been provided before adding Sybil nodes). \vspaceaftercaption \vspace{-3mm} %
    }%
    \label{fig:eclipse_timeline}%
\end{figure}%
When the Sybil nodes are placed before the target content is added to the network, the attack takes effect immediately ($\twarmup=0$). However, we also explore a more realistic timeline in which content has been provided before launching the attack. %
\Cref{fig:eclipse_timeline} presents the evolution of the attack effectiveness $\aeff$ over time. For each censored CID, we spawn $5$ \downloader nodes per hour.
We stop the experiment when all queries have been unsuccessful for at least $3$ hours. 
Immediately after starting the attack, the effectiveness reaches $\aeff = 30\%$. This is caused by a portion of the \downloaders not encountering any honest \resolvers on their path toward the attacked region. The effectiveness steadily increases over time. Multiple spikes (\eg after 12h) are caused by a portion of \resolvers (running older IPFS versions) dropping the \provider records. 
The attack takes full effect after $48$ hours, which is when \resolvers drop the \provider record as per the current IPFS version.
Based on this result, we assume a maximum warmup time $\twarmup=48$h in the cost calculations that follow.

\begin{figure}[t]
    \centering
    \includegraphics[width=\linewidth]{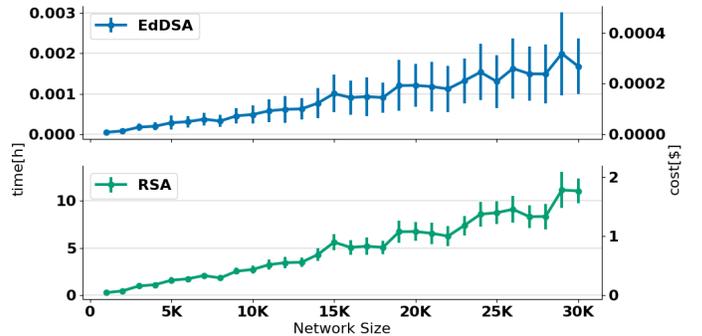}
    \caption{Average time and cost required to generate necessary Sybil identities. \vspaceaftercaption \vspace{-3mm} %
    }
    \label{fig:sybil_generation_time}
\end{figure}

Finally, we explore the \emph{cost} of performing the attack.
The AWS instances running the Sybils cost $\coper=0.16\$$ per hour.
\Cref{fig:sybil_generation_time} presents the time and monetary cost of generating $\numSybils = 45$ Sybil identities for both cryptosystems present in IPFS. For readability, we omit the number of iterations $s$, as it is the same for both methods and proportional to the cost/time. The monetary cost $\cgen$ and the generation time increase linearly with the network size. In larger networks, the distance between the closest honest \resolver and the target CID decreases and the generation algorithm requires more iterations. EdDSA is significantly faster than RSA and the generation time remains below $12$s translating into $0.0005$\$, even for the largest evaluated network with $n=30,000$ nodes. Generating Sybils using RSA, while slower, is feasible even for moderately resourceful attackers.

The peak CPU utilization of $30\%$ occurred only during the generation of Sybil private keys. The maximum bandwidth utilization of the machine hosting $\numSybils=45$ Sybils was $4.67$ Mbps (both inbound and outbound) when no requests were made for the censored CID. Taking the cost of generating Sybil identities $\cgen=0.0005$\$, the longest warmup time $\twarmup=48$h, and the duration of the attack $\teff$ hours, the total cost of the attack using EdDSA is given by $\catt = \cgen  + (\twarmup + \teff) \times \coper  = 7.68 + \teff \times 0.16\$$.

\vspacebeforesection
\subsection{Detection}
\begin{figure}[htbp]
    \centering
    \includegraphics[width=\linewidth]{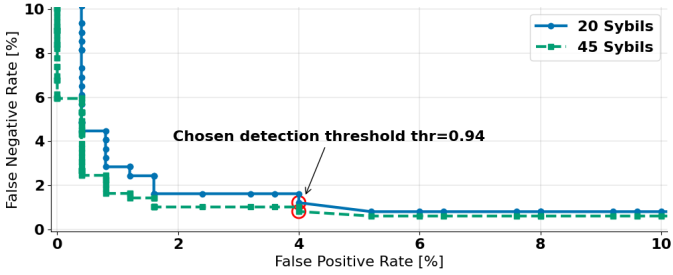}
    \caption{False positive and false negative rates of the detection method, for different detection thresholds. \vspaceaftercaption}%
    \label{fig:fp-fn}
\end{figure}

We follow by investigating our attack detection mechanism, which calculates the KL divergence between a model and empirical peer ID distributions, and flags a CID as under attack if the divergence is above a certain threshold $\threshold$.
In \Cref{fig:fp-fn}, we plot the false positive and false negative rates of the detection method for different choices of the detection threshold (the lower the threshold, the more false positives, but the fewer false negatives).
We see that increasing the number of Sybils makes the attack easier to detect, with fewer false negatives.
Each DHT node can choose its own detection threshold, based on its desired false positive and false negative rates.
However, it is reasonable for the default implementation to choose a threshold that favors fewer false negatives, thereby mitigating most attacks, at the cost of a small overhead of running the mitigation even when there is no attack.
In the following experiments, we choose a threshold of
$0.94$ which achieves $4.4\%$ false positives and $0.81\%$ false negatives
(circled in \Cref{fig:fp-fn}).

\begin{figure}[tb]%
    \centering%
    \includegraphics[width=\linewidth]{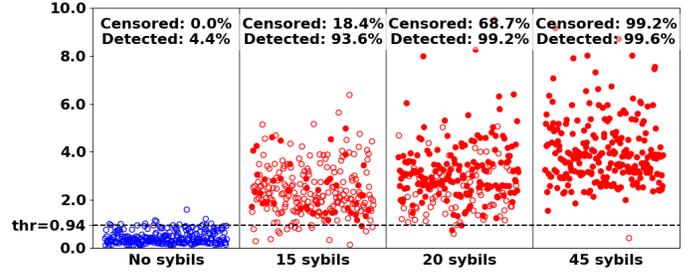}%
    \caption{%
        KL divergence for varying numbers of Sybils $e$. \vspaceaftercaption
    }%
    \label{fig:kl-threshold}%
\end{figure}%

In \Cref{fig:kl-threshold}, we show the results of the detection method for $250$ experiments each with a different number of Sybil peers. Each point represents the result of a single experiment and the y-coordinate is the observed KL divergence value. The percentage of successful attacks (solid circles), where the provider record was not found by the \downloader, and the percentage of experiments that were flagged as attacks are indicated. If the number of Sybils launched by the attacker is decreased, more attempted attacks go undetected (false negatives), but we see that these attacks are not successful either.
This effect is also summarized in \Cref{fig:attack_success_rate}.

\begin{figure}[tb]%
    \centering%
    \includegraphics[width=\linewidth]{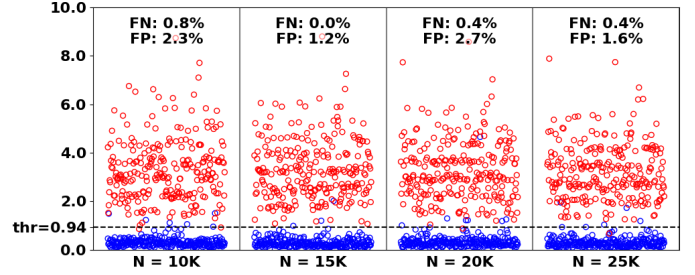}%
    \caption{%
        KL divergence for varying network sizes $N$, and the false negative and false positive rates. \vspaceaftercaption
    }%
    \label{fig:diff_netsizes}%
\end{figure}%

To run experiments for different network sizes $N$, we simulate a DHT network by generating $N$ random peer IDs. Requests for a CID are resolved by simply finding the $20$ closest peer IDs in this simulated network. Network size estimation and detection are performed using these simulated responses. In \Cref{fig:diff_netsizes}, we show that the KL divergence metric is robust to changes in the network size. This is because our detection method automatically calculates the model distribution based on the estimated network size.
Therefore, a fixed detection threshold can be used even as the network size changes.
Note that the KL divergence values in such a simulation tend to be lower than those measured on the real network because lookups in the real network do not always result in the correct $20$ closest peers.

\vspacebeforesection
\subsection{Mitigation}

We evaluate the performance and the overhead of using our mitigation mechanism in the live IPFS network. Mitigation of a censorship attack on $\cid$ is successful if its \downloaders successfully retrieve at least one valid provider record using a $\Call{FindProviders}{\cid}$ operation. Our mitigation mechanism uses region-based queries 
with a region containing $20$ peer IDs on average
for both $\Call{FindProviders}{\cid}$ and $\Call{Provide}{\cid}$ when an attack is detected. %

For different number of Sybils $\numSybils$, we launch censorship attacks on 50 different CIDs. During each attack, we launch $\numSybils$ Sybil peers, and after waiting for one minute, we provide the target CID from a separate DHT instance. We then test the reachability of the content from ten different \downloaders. %
As a baseline comparison, we also provide results for $\numSybils=0$ whenever appropriate. 

 \begin{figure}[b]
    \centering
    \includegraphics[width=\linewidth]{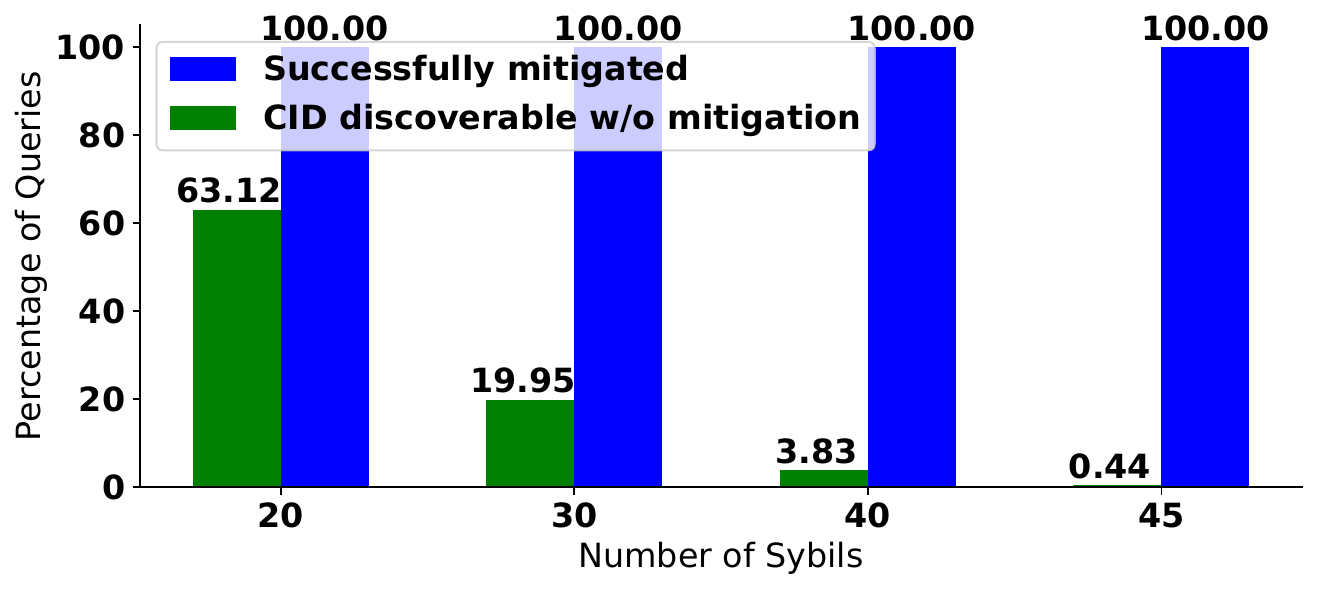}
    \caption{Percentage of attacks that are mitigated. \vspaceaftercaption}
    \label{fig:mit-success-sybils}
\end{figure}

\para{Mitigation effectiveness} \Cref{fig:mit-success-sybils} illustrates the mitigation effectiveness $\meff$ for different numbers of Sybils $e$, and provides results without the mitigation for comparison. %
Our mechanism mitigates all of the detected attacks for all the evaluated Sybil numbers. Importantly, for $e=45$, the number of \downloaders receiving their content increases from only $0.44\%$ without the mitigation to $100\%$ when it is activated.

 \begin{figure}[t]
    \centering
    \includegraphics[width=\linewidth]{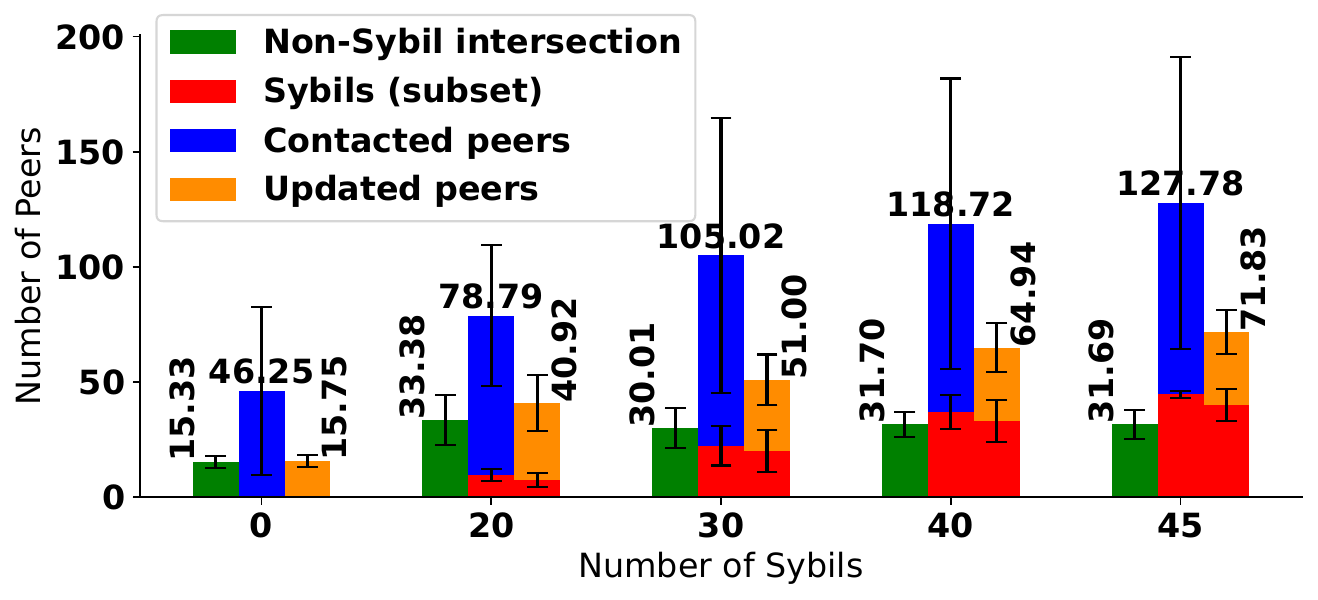}
     \caption{Number of peers contacted and updated by region-based queries. \vspaceaftercaption}%
     \label{fig:mit-overhead-sybils}
    \end{figure}

To better understand the region-based queries, \Cref{fig:mit-overhead-sybils} presents the average size of the following sets of peers: (i) \textit{contacted:} peers encountered by the \downloaders during \Call{FindProviders}{\cid}, (ii) \textit{updated:} \resolvers that obtain the provider record for $\cid$, and (iii) \textit{intersection:} non-Sybil peers in the intersection of the two aforementioned sets of peers. We omit unresponsive peers in both sets. The number of contacted peers increases with the increasing number of Sybils. Higher peer density in the target region causes additional lookups that reach honest peers located nearby. For the same reason, the number of updated peers increases as well.

The number of contacted peers is substantially higher than the number of updated peers, because the former includes not only just the peers within the target region but also peers encountered during the DHT walk toward that target region. More importantly, the average size of the non-Sybil intersection of contacted and updated peers oscillates around 30. This is higher than the expected value of 20 (\ie, the region size used by the mitigation mechanism) because the current network size estimation slightly underestimates the size of the IPFS network. Higher intersection size ensures successful mitigation even when multiple honest nodes in the region are offline.

\begin{figure}[b]
    \centering
    \includegraphics[width=\linewidth]{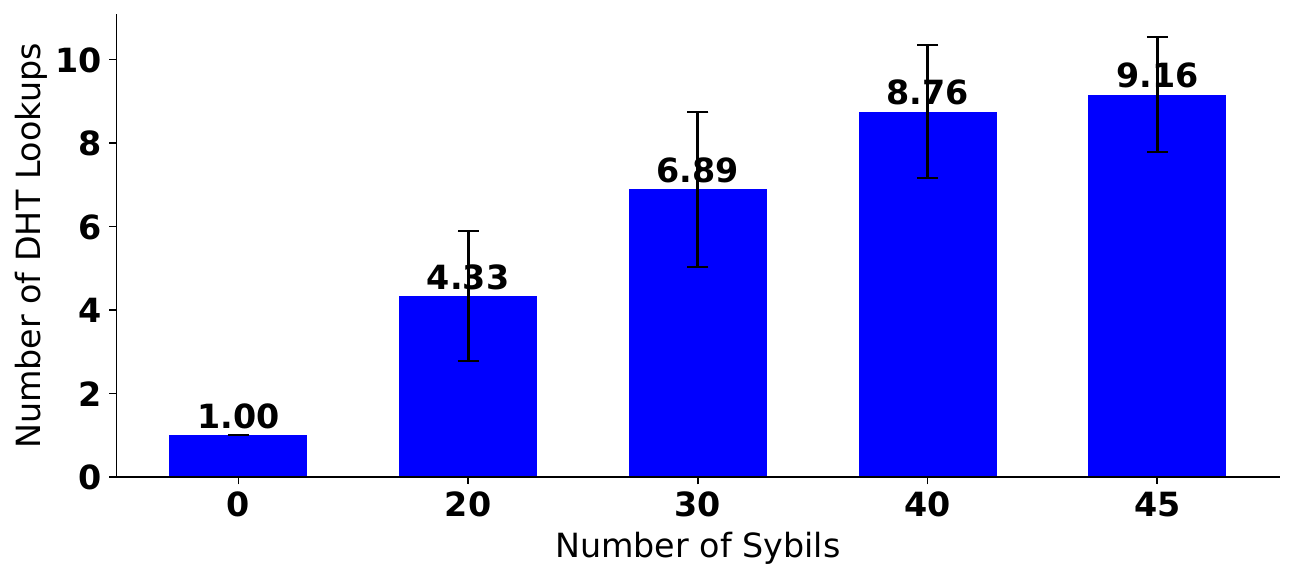}
     \caption{The number of DHT lookups involved in a region-based query. \vspaceaftercaption}
    \label{fig:mit-lookups-sybils}
\end{figure}

\para{Mitigation overhead} In \Cref{fig:mit-lookups-sybils}, we assess the overhead of the mitigation mechanism in terms of the number of DHT lookups involved in region-based queries. While the number of lookups is up to $9$ times higher than for an un-attacked network, the overhead increases sub-linearly with the number of Sybils, and is incurred only when an attack is detected.

We then measure the latency of the $\Call{Provide}{}$ as well as $\Call{FindProviders}{}$ operations. The average latencies of these operations are shown in \Cref{tab:latency} when under attack and otherwise, for both the default and the mitigation modes.
In general, $\Call{FindProviders}{}$ waits until it finds $20$ distinct providers, or it has contacted all \resolvers. This is followed in both the default and mitigation modes. However, the first \provider record is obtained much earlier. Although finding one \provider record may be enough in the optimistic case, the first \provider record that the \downloader obtains may not be correct (it may be old or may be sent by a malicious peer) and therefore the \downloader might need to wait longer to retrieve the content.
Even though the mitigation mode requires several DHT queries while the default mode uses only one, the latency of the mitigation mode is not much higher than the default mode. This is because the subsequent queries are much faster than the first query.

\begin{table}[t]
    \newcolumntype{G}{>{\raggedright\arraybackslash} m{0.2\linewidth} }
    \newcolumntype{H}{>{\raggedright\arraybackslash} m{0.15\linewidth} }
    \newcolumntype{I}{>{\raggedright\arraybackslash} m{0.227\linewidth} }
    \renewcommand{\arraystretch}{1.2}

    \begin{tabular}{GHII}
    
    \toprule
    & & Attack (45 Sybils) & No attack (0 Sybils) \\
    \midrule
    \multirow{2}{*}{Provide} & Default & \textcolor{red}{---} & 15914 ms \\
    & Mitigation & 24780 ms & 27099 ms \\
    \midrule
    \multirow{2}{*}{Find $20$ providers} & Default & \textcolor{red}{---} & 26483 ms \\
    & Mitigation & 28930 ms & 27756 ms \\
    \midrule
    \multirow{2}{*}{Find $1$ providers} & Default & \textcolor{red}{---} & 647 ms \\
    & Mitigation & 712 ms & 329 ms \\
    \bottomrule
    
    \end{tabular}
    \caption{Average latency (milliseconds) of Provide and FindProviders operations with and without our mitigation, during attack and no attack. A red dash indicates that the operation was unsuccessful due to the attack. \vspaceaftercaption}
    \label{tab:latency}
\end{table}

%

%% file: sections/related.tex
\vspacebeforesection
\section{Related Work}
\label{sec:related}

In this section, we review previous work focusing on solving the problem of attacks based on Sybil identities in decentralized systems and their limitations. For a more complete view of the Sybil attack and countermeasures, we refer the readers to a survey by Urdaneta \etal~\cite{survey11}.

CFS~\cite{dabek2001wide} is a storage system built over the early Chord DHT~\cite{stoica2003chord}. CFS uses node ID authentication to prevent a node from taking a specific position in the DHT ring.
CFS clients check that the node ID is the result of the hash of its IP address, plus a number from a small range, \eg, 1 to 10. However, this solution is less effective when an attacker has access to a large number of IP addresses (\eg using cloud providers) and it is incompatible with a large number of peers placed behind a NAT, which is the case in IPFS~\cite{trautwein2022design}. S-Chord~\cite{byzantine_chord} is an extension of Chord that can provide routing guarantees despite the presence of a number of Byzantine nodes in the network but it increases the number of messages and latency for routing by a factor logarithmic in the number of nodes. S/Kademlia~\cite{baumgart2007s} proposes Proof-of-Work (PoW) mechanisms to rate-limit the generation of new peer IDs. However, while PoW slows down the attacker, it does not fully mitigate the problem, makes the system less sustainable, and is problematic for constrained devices.

Some mechanisms make additional assumptions on trusted certificate authorities (CAs) to sign peer IDs~\cite{castro_dht} or use social trust networks~\cite{sybillimit, sybilguard, danezis2005sybil, whanau} to detect or prevent Sybil attacks.
Even though most deployed systems do rely on hardcoded bootstrap nodes, relying on CAs to control and certify all memberships would be considered incompatible with the open and decentralized environment of IPFS. 

Awerbuch and Scheideler~\cite{awerbuch2009} propose that the peer IDs of all honest peers in a DHT be rotated whenever a new peer joins. This can prevent an attacker's peers from concentrating in one region of the key space. Unfortunately, this solution is particularly expensive in a dynamic network where nodes constantly join and leave the system.

Cholez \etal~\cite{cholez2010detection} introduce a Sybil detection mechanism based on KL-divergence followed by removing suspected peers from the set of $k$ closest peers.
We adopt their detection mechanism but do not remove any peers.
Since in IPFS, the CID is the hash of the content, Sybil peers cannot cause a downloader to accept incorrect content. Therefore, removing Sybil peers is not required. Instead, our mitigation ensures that providers and downloaders continue to contact enough honest peers.

Recently, Protocol Labs introduced network indexers~\cite{ipfsindexer} allowing to resolve a CID to a list of providers in Filecoin~\cite{fisch2018scaling}. While the usage of cloud infrastructure makes the system highly efficient and resistant to Sybil attacks, the indexer is fully centralized introducing the risks of censorship and can constitute a single point of failure.

\para{Eclipse attacks}
Multiple attacks based on node \emph{eclipsing} target decentralized systems.
An attacker attempts to control all the neighbors of a specific target node in the overlay.
This differs from the content censorship attack we discuss in this paper, which targets a specific entry of a distributed directory.

Eclipse attacks are documented for Bitcoin~\cite{heilman2015eclipse,saad2021syncattack} or Ethereum~\cite{henningsen2019eclipsing, marcus2018low} allowing to \emph{partition} the blockchain network and prevent a miner from participating fairly.
The recent Gethlighting attack shows this is possible by only eclipsing a subset of a node's neighborhood~\cite{heopartitioning}.
A recent attack targets the IPFS DHT~\cite{henningsen2020mapping, total_eclipse} and allows isolating a single node from the network.
As a result, the IPFS DHT was augmented with table eviction policies and with rules restricting the number of peers with the same IP address in a routing table.
This makes the attack impractical even for a resourceful attacker~\cite{total_eclipse}.
Our censorship attack targets content rather than single nodes and works despite these changes. However, we also build upon this past work as we rely on eclipse resistance for our mitigation techniques.

Wang \etal~\cite{wang2008attacking} is an early example of a content censorship attack, targeting the Kad network, an implementation of Kademlia used in the eDonkey~\cite{edonkey} and eMule~\cite{emule} content-sharing networks.
These attacks exploit a vulnerability in the Kad implementation: peers were not authenticated based on their peer IDs, so an attacker could impersonate another peer.
This vulnerability does not exist in the IPFS network where peer IDs are derived by hashing the peer's public key, and where messages are signed with the corresponding secret keys.
Our attack is much simpler than the one of Wang \etal and does not require this vulnerability.

%% file: sections/discussion.tex
\section{Discussion and Future Work}\label{sec:discussion}

Currently, IPFS does not provide an admission mechanism and \resolvers will accept any provider records until they run out of storage. After that, depending on the implementation, the \resolvers may crash or flush older, legitimate provider records. The time required for this attack depends on the bandwidth available at each \resolver and the amount of free storage. While a deeper analysis is out of the scope of this paper, introducing an admission mechanism based on the diversity of incoming traffic has the potential to eliminate this vulnerability.

While our mitigation technique (\Cref{sec:mitigation}) fully protects against the CID censorship attack, it involves querying the Sybil nodes for provider records. The Sybil nodes may return a large number of fake provider records so that \downloaders keep trying them, thereby slowing down the resolution. The impact of such an attack can be reduced if the \downloader only tries a single provider record obtained from each \resolver and prefers records obtained from \resolvers with diverse IP addresses (\ie from different /24 networks). Any attempts to significantly delay the resolution would sharply increase the attacker's cost. 

Our mitigation technique relies on region-based DHT queries. For easy integration with the current IPFS network, those queries are built on top of a regular Kademlia DHT that does not natively support them. This results in slightly higher overhead and increases resolution time. Adapting the core internals of the DHT and optimizing them for region-based queries might speed up the process and reduce its overhead. However, such deep changes make incremental deployment and compatibility with the existing version challenging. 

During this project, we initially considered an approach where \providers register their provider records on all the nodes encountered on the path towards \resolvers. Such a solution increases the chance of an honest \downloader receiving a correct provider record before reaching the region with the Sybil nodes. However, the mechanism does not provide resistance for \downloaders located close to the CID in the DHT hash space. Furthermore, the on-path registration significantly increases the storage cost of holding provider records for the entire network even when no attack is being conducted. 

The IPFS DHT, and thus our mitigation and detection mechanisms, depends on the correctness of the DHT routing. However, a powerful attack may try to disturb DHT operations by deploying a large number of uniformly distributed Sybils that only return other Sybils when queried. While costly, such an attack could be devastating for the entire ecosystem. We advocate for additional future work that improves the DHT resistance to such attacks and is practical to deploy in large-scale networks.

%% file: sections/conclusion.tex
\vspacebeforesection
\section{Conclusion}
\label{sec:conclusion}

We presented a successful censorship attack on IPFS. We showed that an attacker can easily make  any content undiscoverable in the network by strategically placing a small number of Sybil identities in the DHT. The effectiveness of the attack was confirmed by removing multiple, specifically crafted content from the live IPFS network. Importantly, our attack has a constant, negligible cost regardless of the popularity of the target content.

The attack has a significant impact on the IPFS network itself as it threatens the core functionality of the platform. However, it also impacts other systems that rely on the availability of content stored on IPFS. This includes thousands of decentralized applications and oracles deployed on various blockchains. Moreover, the DHT flaw that led to the attack is present in other currently deployed DHT-based systems.

We also presented a robust detection technique allowing us to detect the attack in real time without communication overhead and to activate our proposed mitigation mechanisms when necessary. 

Finally, we introduce a practical mitigation technique based on region-based DHT queries. While many others mitigation techniques have been proposed, none of them are practical enough to be deployed in an open decentralized system. Our approach is the first that can be deployed incrementally in a live network without requiring changes to the core DHT protocol.
It also does not require additional components and does not incur significant overhead. Importantly, our mitigation technique prevents the attack without blocking any nodes or using unreliable reputation systems. We believe that our mitigation technique can be easily integrated into other DHT-based systems.

%% file: sections/ae_appendix.tex
\section{Artifact Appendix}
\label{sec:ae-appendix}

In this work, we implement a censorship attack on the IPFS network, a method to detect the attack, and a method to mitigate the attack.
Our artifact includes the implementations of these three components and experiments to measure their effectiveness, accuracy, and cost.

\smallskip
\noindent
Artifact Outline (key aspects):

{
\footnotesize
\dirtree{%
.1 /.
.2 README.md.
.2 common.
.3 go-libp2p-kad-dht \begin{minipage}[t]{5cm}
(detection \& mitigation)
\end{minipage}.
.3 go-libp2p-kad-dht-Sybil \begin{minipage}[t]{5cm}
(Sybil DHT implementation)
\end{minipage}.
.3 Sybil\_DHT\_Nodes \begin{minipage}[t]{5cm}
(Sybil node implementation)
\end{minipage}.
.3 kubo \begin{minipage}[t]{5cm}
(standard IPFS implementation)
\end{minipage}.
.2 experimentCombined \begin{minipage}[t]{5cm}
(main experiment \& results)
\end{minipage}.
.3 README.md.
.2 python \begin{minipage}[t]{5cm}
(plots \& other simulations)
\end{minipage}.
.3 README.md.
}
}

\subsection{Description \& Requirements}
\label{sec:ae-appendix-requirements}

\subsubsection{How to access}
\label{sec:ae-appendix-access}
The artifact is available online
at the link \url{\artifactlink}.
The detection and mitigation parts are also available on Github: \url{\githublink},
and
are scheduled to be deployed in the official release of \path{go-libp2p-kad-dht}.
Our mitigation will be deployed on a per-client basis (not network-wide) and therefore, the experiments in this artifact are expected to remain reproducible as they attack our own \providers and \downloaders that do not use the mitigation.

\subsubsection{Hardware dependencies}
\label{sec:ae-appendix-hardware}
Our experiments require a machine with a public IP address which must allow incoming TCP connections on several ports to allow other IPFS peers to connect to our Sybil peers.
Our experiments were run on a machine rented from AWS.
The recommended instance type is \verb|t3.xlarge| which has a 2nd generation Intel Xeon Scalable Processor (3.1 GHz) with 4 vCPUs, 16 GB memory, and 5 Gbps peak bandwidth.
While the compute and memory requirements of our experiments are modest, a high peak bandwidth is required to ensure good connectivity of the Sybils and good attack effectiveness.
Our artifact includes instructions in the README on how to run the experiments on any other machine with these requirements.

\subsubsection{Software dependencies} 
\label{sec:ae-appendix-software}
Recommended operating system: Ubuntu 22.04. Required software: Go v1.19.10, Python 3.10, gcc, Make.

\subsubsection{Benchmarks} 
None.

\subsection{Artifact Installation \& Configuration}
\label{sec:ae-appendix-installation}

For the artifact evaluation, the setup begins with logging in to our provided AWS machine via SSH using the instructions given in the README file. Detailed instructions for each of the following steps are given in the top-level README file.

Network setup: In firewall settings, incoming TCP connections on ports 4001, 5001, 63800-63850 must be enabled.

Install software requirements: Install Go, Python, gcc, and make. We provide a file requirements.txt to help install all required Python packages.

Build and initialize IPFS: We provide the source code of kubo, the IPFS client written in Go, in our artifact. After compiling the source code, an IPFS node must be initialized.
The IPFS node runs in the background during our experiments and is used to obtain information from the network such as the closest peers to the target CID.

\subsection{Experiment Workflow}
\label{sec:ae-appendix-workflow}

Our main experiment (E1) has two phases: data collection and plotting.

Data collection: In this phase, we run the attack for different numbers of Sybils, different target CIDs, and different \downloader clients. We collect measurements regarding the attack's success, detection results, and mitigation success each time. The code for this step along with instructions is given in the folder \path{experimentCombined/}.

Data processing and plotting: In this phase, we process the collected measurements and generate the plots in the paper (\Cref{fig:attack_success_rate,fig:fp-fn,fig:kl-threshold,fig:mit-success-sybils,fig:mit-lookups-sybils,fig:mit-overhead-sybils}). The code and instructions for this step are in the folder \path{python/}.

\subsection{Major Claims}
\label{sec:ae-appendix-claims}

\begin{itemize}
    \item (C1): Our censorship attack using $e=45$ Sybil nodes blocks $99\%$ of users' content requests. This is illustrated in \Cref{fig:attack_success_rate}.
    \item (C2): Our detection mechanism achieves a false negative rate of 0.81\% and a false positive rate of 4.4\% for our chosen detection threshold (\Cref{fig:fp-fn}). Most attacks that were not detected were also not successful in censoring the content, and the detection rate improves as the attack effectiveness increases (\Cref{fig:kl-threshold}). 
    \item (C3): Our mitigation leads to successful content discovery 100\% of the time (\Cref{fig:mit-success-sybils}) by sending \provider records to a constant number of honest peers even as the number of Sybils increases (\Cref{fig:mit-lookups-sybils}). Moreover, the overhead of our mitigation increases sub-linearly with the number of Sybils (\Cref{fig:mit-overhead-sybils}). Claims (C1), (C2), and (C3) are proven by experiment (E1).
    \item (C4): Our censorship attack incurs a low cost for the attacker. The required keys can be generated on commodity hardware within $20$ seconds (for EdDSA) or 2 hours (for RSA). This is proven by the experiment (E2) whose results are illustrated in \Cref{fig:sybil_generation_time}.
\end{itemize}

\subsection{Evaluation}
\label{sec:ae-appendix-evaluation}

\subsubsection{Experiment (E1)}
[Detection and Mitigation] [4 compute-hours]:
In a single run of this experiment, we do the following: We create a new file, compute its CID, generate Sybil identifiers based on the CID, and launch Sybil nodes. Then, we launch two DHT nodes, a \provider and a \downloader. The \provider \emph{provides} the file, the \downloader attempts to find \providers for the file and runs the detection. We record the success of the attack and the detection result. Then, the \provider provides the file with the mitigation enabled, the \downloader attempts to find \providers also with the mitigation enabled, and we record the mitigation success, the number of DHT lookups performed by the mitigation, and the number of honest and Sybil peers contacted and successfully updated with \provider records. This whole process is repeated for 
10 different files.
This experiment is then repeated for 0, 20, 30, 40, and 45 Sybils. Each run takes 4-5 minutes on average. Hence, we scale down the number of runs for each experiment from 100 to 10 and run it for fewer values of the number of Sybils so that the experiment completes within 4 hours.

\textit{[Preparation]}
None beyond \Cref{sec:ae-appendix-installation}.

\textit{[Execution]}
Change the working directory to \path{experimentCombined/}, then build and run the experiment as per instructions in \path{experimentCombined/README.md}. Running the experiment stores the results in the specified output directory.

\textit{[Results]}
The folder \path{python/} contains code to generate \Cref{fig:attack_success_rate,fig:fp-fn,fig:kl-threshold,fig:mit-success-sybils,fig:mit-lookups-sybils,fig:mit-overhead-sybils} from the experiment results. For reference, we also provide the results corresponding to the plots in the paper in \path{experimentCombined/detection_results} and \path{experimentCombined/mitigation_results}.

\subsubsection{Experiment (E2)}
[Keys Generation Time] [10 compute-hours]: This experiment generates the required Sybil keys for different sizes of the network and analyzes the generation time.

\textit{[Execution]}
Run two sets of experiments measuring generation time for RSA and EDDSA:
\begin{lstlisting}[language=bash,basicstyle=\footnotesize]
cd python/
go run measureGenerateSybilKeys rsa > timing_rsa.csv
go run measureGenerateSybilKeys eddsa > timing_eddsa.csv
\end{lstlisting}
We provide our results files for reference (generating RSA keys takes a while).
\begin{lstlisting}[language=bash,basicstyle=\footnotesize]
./simulation_results/sample_rsa.csv
./simulation_results/sample_eddsa.csv
\end{lstlisting}

\textit{[Results]}
The commands below will reproduce the graph in \Cref{fig:sybil_generation_time}.
\begin{lstlisting}[language=bash,basicstyle=\footnotesize]
python3 plot_key_generation_time.py
\end{lstlisting}

\subsection{Customization}
Experiment (E1) provides arguments to customize the number of Sybils, the number of clients, and the number of CIDs for which to run the experiment.

\subsection{Notes}

The key generation time is heavily hardware dependent. The result might thus differ from the ones shown in \Cref{fig:sybil_generation_time}.
The experiment time might be shortened by adjusting the number of tries per network size. However, this comes at the price of losing the precision.

In the artifact, we include the code and instructions for generating all other figures in the paper as well.